\documentclass[12pt]{article}
\usepackage{mathtools,amssymb,amsthm}
\usepackage{newtxtext,newtxmath,bm,pifont}
\usepackage{cases}
\usepackage{mathrsfs}
\usepackage{graphicx}
\usepackage[bookmarks=true,bookmarksnumbered=true,colorlinks=true]{hyperref}
\usepackage{color}
\usepackage[tight]{subfigure}
\usepackage{booktabs}
\usepackage{multirow}
\usepackage{makecell}
\AtBeginDvi{}
\theoremstyle{plain}
\newtheorem{thm}{Theorem}[section]
\newtheorem{prop}[thm]{Proposition}

\newtheorem{cor}[thm]{Corollary}

\newtheorem{rem}[thm]{Remark}

\newcommand*{\tr}[2][-3mu]{\ensuremath{\mskip1mu\prescript{\smash{\mathrm
        t\mkern#1}}{}{\mathstrut#2}}}

\makeatletter
\@addtoreset{equation}{section} 

\makeatother

\usepackage[normalem]{ulem}
\usepackage{cancel}

\makeatletter
\newcommand*{\diff}{\@ifnextchar^{\DIfF}{\DIfF^{}}}
\def\DIfF^#1{\mathop{\mathrm{\mathstrut d}}\nolimits^{#1}\gobblespace}
\def\gobblespace{\futurelet\diffarg\opspace}
\def\opspace{\let\DiffSpace\!\ifx\diffarg(\let\DiffSpace\relax\else\ifx\diffarg[\let\DiffSpace\relax\else\ifx\diffarg\{\let\DiffSpace\relax\fi\fi\fi\DiffSpace}

\newcommand*{\pderiv}[3][]{\frac{\partial^{#1}#2}{\partial #3^{#1}}}

\newcommand*{\R}[1][] {\mathbb{R}^{#1}}
\newcommand*{\N}[1][] {\mathbb{N}^{#1}}
\newcommand*{\norm}[1]{\left\lVert#1\right\rVert}

\newcommand*{\ord}[1]{\mathcal{O}(#1)}
\newcommand*{\order}[1][]{\ord{h^{#1}}}
\DeclareMathOperator*{\argmin}{arg\,min}
\DeclareMathOperator{\sn}{sn}
\DeclareMathOperator{\cn}{cn}
\DeclareMathOperator{\dn}{dn}

\DeclareMathOperator{\st}{|}
\newcommand*{\braket}[2]{\left<{#1},{#2}\right>}

\begin{document}
\begin{center}
\begin{large}
Fairing of Discrete Planar Curves by Integrable Discrete Analogue of Euler's Elasticae \\[5mm]  
\end{large}
\begin{normalsize}
Sebasti\'an El\'ias {\sc Graiff Zurita}\\
Graduate School of Mathematics
Kyushu  University\\
744 Motooka, Nishi-ku, Fukuoka 819-0935, Japan\\[2mm] 
Kenji {\sc Kajiwara}\\
Institute of Mathematics for Industry, Kyushu University\\
744 Motooka, Fukuoka 819-0395, Japan\\[2mm]
Toshitomo {\sc Suzuki}\\
Department of Architecture, Mukogawa Woman's University\\
1-13 Tozaki-cho, Nishinomiya, Hyogo 663-8121
\end{normalsize}
\end{center}

\begin{abstract}
  We construct a method to fair a given discrete planar curve by using the integrable discrete
  analogue of Euler's elastica, which is a discrete version of the approximation algorithm presented
  by D. Brander, et al.  We first give a brief review of the integrable discrete analogue of Euler's
  elastica proposed by A. I. Bobenko and Yu. B. Suris, then we present a detailed account of the
  fairing algorithm, and we apply this method to an architectural problem of characterizing the
  keylines of Japanese handmade pantiles.
\end{abstract}


\section{Introduction} \label{sec:cont.planar}
The Euler's elastica (elastic curve) is a class of planar curves characterized as the solutions to
the variational problem of minimizing the elastic energy under a certain boundary condition.  It has
been regarded as one of the most important class of planar curves because it is endowed with rich
mathematical structure:  exact solutions, integrability, geometry of elliptic curves, and so on,
while it serves as a simple but realistic model of thin inextensible elastic rod (see, for example,
\cite{Matsutani:elastica,Singer:Lecture}). Brander et al. \cite{brander} have proposed an algorithm
to fair a given planar curve segment by an Euler's elastica, motivated mainly by the development of
the robotic hot-blade cutting technology.  In this work, motivated by a problem of architecture to
characterize the keylines of Japanese handmade pantiles, where the curve data is obtained in the
form of discrete point data, we aim to construct a fairing method of discrete planar curves by using
the integrable discrete analogue of the Euler's elastica proposed by Bobenko and Suris
\cite{bobenko.1999}, which is referred to as the {\em discrete elastica} in this paper.

This paper is organized as follows. We give a brief review of the Euler's elastica and the discrete
elastica in Sections \ref{sec:elastica} and \ref{sec:disc.planar}, collecting the information on
variational formulations, exact solutions and continuum limits with proofs.  We present a detailed
account of the fairing of a given discrete planar curve by the discrete elastica in Section
\ref{sec:approx.elastica}. Finally, application to the characterization of the keylines of Japanese
handmade pantiles is discussed in Section \ref{sec:application}. For various formulas of the
Jacobi elliptic functions used in this paper, the readers may refer to \cite{DLMF}, for example.
%
\section{Euler's elastica}\label{sec:elastica}
Let $\gamma(s) \in \R[2]$ ($s \in \R$) be an arc length parameterized planar curve.  By definition,
it holds that $\norm{\gamma'(s)} = 1$, where ${}' = \frac{d}{ds}$.  The tangent and normal vectors
are defined by $T(s) = \gamma'(s)$ and $ N(s) = R(\pi/2)T(s)$, respectively, where
\begin{equation}
  R(\varphi)=\begin{bmatrix}\cos\varphi & -\sin\varphi\\\sin\varphi &\cos\varphi\end{bmatrix}.
\end{equation}
Due to $\norm{\gamma'(s)} = 1$, it is possible to parameterize
the tangent vector as
\begin{equation}
 T(s) = \begin{bmatrix} \cos\theta(s)\\\sin\theta(s)\end{bmatrix},
\end{equation}
where the {angle function} $\theta(s)$ is the angle of $T(s)$ measured from the horizontal axis
in the counterclockwise direction.  Introducing the {Frenet frame} $\Phi(s)$ by
\begin{equation}
 \Phi(s) =[T(s),N(s)]\in\text{SO}(2),
\end{equation}
we have the {Frenet formula},
\begin{equation}\label{eqn:Frenet_formula}
\Phi'(s) = \Phi(s) L(s),\quad L(s)=\begin{bmatrix}0 & -\kappa(s)\\ \kappa(s) & 0 \end{bmatrix},
\end{equation}
where $\kappa(s) = \theta'(s)$ is the (signed) curvature.
%
\begin{figure}[h]
\begin{center}
\includegraphics[scale=0.6]{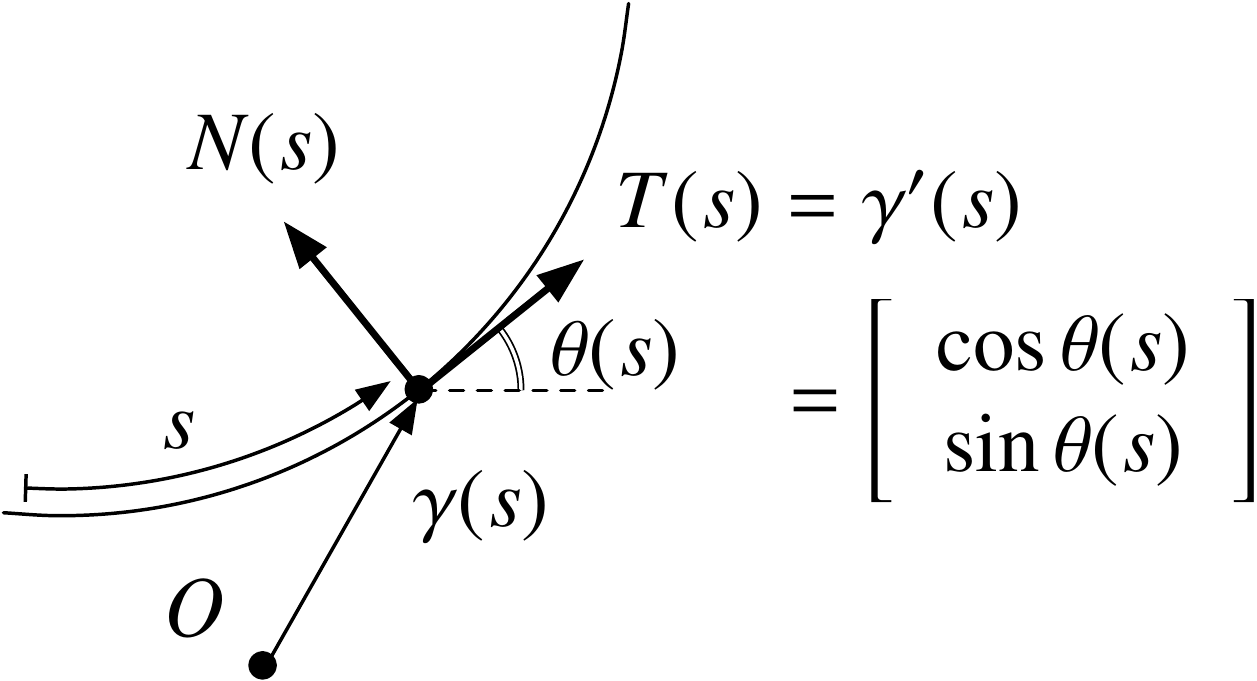}
\caption{Smooth planar curve and the Frenet frame.}  
\end{center}
\end{figure}
%
The {Euler's elastica} (or simply referred to as the {elastica}) is defined as a critical
point of the {elastic energy}
\begin{equation} \label{eqn:elastic.energy}
 E = \int_0^L \big(\kappa(s)\big)^2\,ds,
\end{equation}
with respect to variations with fixed endpoints and tangent vectors at the endpoints, under the
condition of preserving the total length. The Euler-Lagrange equation yields the following
differential equations for the curvature and the angle function:
%
\begin{prop}\label{prop:Euler-Lagrange}
The curvature $\kappa$ of the Euler's elastica satisfies
\begin{equation}\label{eqn:elastica_ODE_kappa}
 \kappa'' + \frac{1}{2}\kappa^3 - \lambda\kappa = 0,
\end{equation} 
where $\lambda\in\mathbb{R}$ is a constant. Moreover, the angle function $\theta$ satisfies
\begin{equation}\label{eqn:elastica_ODE_theta}
 \theta'' +\mu\sin\theta=0,
\end{equation}
where $\mu>0$ is a constant.
\end{prop}
Derivation of equation \eqref{eqn:elastica_ODE_kappa} is given in various literatures such as
\cite{Singer:Lecture}.  Here, we show a concise derivation using the variation of the tangent
vector.  Consider the functional
\begin{equation}
  S = \int_0^L \Big(\langle T',T'\rangle + c\langle T,T\rangle + \langle a,T\rangle\Big)\, ds,
\end{equation}
where the first term is the elastic energy, and the second and third terms correspond to the
preservation $\norm{T}$ and $\gamma(L) - \gamma(0)$, respectively, with $c = c(s)\in\mathbb{R}$ and
$a\in\mathbb{R}^2$ being the Lagrange multipliers. The variation of $S$ is calculated by using
the Frenet formula \eqref{eqn:Frenet_formula} as
\begin{align}
  \delta S &= \int_0^L \Big(2\langle T',\delta T'\rangle + 2c(s) \langle T,\delta T\rangle + \langle a,\delta T\rangle\Big)\, ds \nonumber\\
           &=2\langle T',\delta T\rangle\Big|_0^L + 
             2\int_0^L \big\langle (\kappa^2+c(s))T  -\kappa'N + \frac{a}{2}, \delta T\big\rangle\, ds.
\end{align}
The first term is the boundary term that vanishes due to the boundary condition, and the second
term gives the Euler-Lagrange equation,
\begin{equation}
(\kappa^2+c)\,T  -\kappa'N + \frac{a}{2}=0.
\end{equation}
Taking the scalar product with $T$ and $N$, we have 
\begin{gather}\label{eqn:EL_1}
\kappa^2 + c + \frac{1}{2}\langle a,T\rangle=0 \quad\text{and}\quad  -\kappa' + \frac{1}{2} \langle a,N\rangle=0,
\end{gather}
respectively.  Multiplying $\kappa'$ to both sides of the second equation of equation
\eqref{eqn:EL_1}, we find that it is integrated to give
\begin{equation}\label{eqn:EL_2}
\frac{\kappa^2}{2} - \lambda =  \frac{1}{2}\langle a,T\rangle,
\end{equation}
where $\lambda \in \R$ is a constant of integration.  Eliminating $\langle a,T\rangle$ from the
first equation of \eqref{eqn:EL_1} and equation \eqref{eqn:EL_2}, $c$ is determined consistently
as $c=-\frac{3}{2}\kappa^2+\lambda$. Differentiating the second equation of equation
\eqref{eqn:EL_1} gives
\begin{equation}\label{eqn:EL_3}
 \kappa''=-\frac{\kappa}{2}\langle a,T\rangle.
\end{equation}
Then eliminating $\langle a,T\rangle$ from equations \eqref{eqn:EL_2} and \eqref{eqn:EL_3} yields
\begin{equation}
 \kappa'' + \kappa\Big(\frac{\kappa^2}{2}-\lambda\Big)=0,
\end{equation}
which is nothing but equation \eqref{eqn:elastica_ODE_kappa}.
%
\begin{rem}\hfill
\begin{enumerate}
 \item Equation \eqref{eqn:elastica_ODE_kappa} is derived from equation \eqref{eqn:elastica_ODE_theta} as follows.
Multiplying $\theta'$ on both sides of equation \eqref{eqn:elastica_ODE_theta}, we see that 
equation \eqref{eqn:elastica_ODE_theta} is integrated to give
\begin{equation}\label{eqn:1st_integral_theta}
 \frac{1}{2}(\theta')^2 = \mu\cos\theta+\lambda,
\end{equation}
where $\lambda$ is a constant of integration.  Then differentiating equation
\eqref{eqn:elastica_ODE_theta} and using equation \eqref{eqn:1st_integral_theta} yields
\begin{equation}
 \kappa'' = -\mu\cos\theta\times\theta' = \Big(-\frac{1}{2}\kappa^2 + \lambda\Big)\kappa,
\end{equation}
which is nothing but equation \eqref{eqn:elastica_ODE_kappa}. 
 \item Equations \eqref{eqn:elastica_ODE_kappa} and \eqref{eqn:elastica_ODE_theta} can be seen as
       travelling-wave reductions of the (focusing) modified KdV equation 
\begin{equation}
 \frac{\partial \kappa}{\partial t} + \frac{3}{2}\kappa^2\frac{\partial \kappa}{\partial s} 
+ \frac{\partial^3 \kappa}{\partial s^3}=0
\end{equation}
and the sine-Gordon equation, 
\begin{equation}
 \frac{\partial^2\theta}{\partial s\partial y}=\sin\theta,
\end{equation}
respectively, where the former describes the integrable deformation of planar curves
\cite{Goldstein-Petrich,kajiwara}.
\end{enumerate}
\end{rem}
%
It is known that the differential equations \eqref{eqn:elastica_ODE_kappa} and
\eqref{eqn:elastica_ODE_theta} can be solved in terms of the Jacobi elliptic functions. This is
well-known, but we present the solutions for the readers' convenience. In the literature, the
solutions are often constructed from the the first integral of \eqref{eqn:elastica_ODE_kappa} given
by
\begin{equation}
 (\kappa^\prime)^2 + \frac{\kappa^4}{2}-\lambda\kappa^2=C,
\end{equation}
where $C$ is a conserved quantity (constant). Here we present those solutions and verify them by the
differential equations for the Jacobi elliptic functions.  
%
\begin{prop}
The curvature $\kappa$ and the angle function $\theta$ can be expressed in terms of the Jacobi
elliptic functions as follows:\\[2mm] 
(i)
\begin{align}
&\kappa = 2 k^{-1} \sqrt{\mu}\,\dn\,\Big(k^{-1}\sqrt{\mu}s; k\Big), \label{eqn:kappa_dn}\\[2mm]
&\sin\frac{\theta}{2} = \sn\Big(k^{-1}\sqrt{\mu}s;k\Big), \quad
\mu=\frac{\lambda}{2k^{-2}-1}.\label{eqn:theta_dn}
\end{align}
(ii) 
\begin{align}
&\kappa = 2 k \sqrt{\mu}\,\cn\,\Big(\sqrt{\mu}s; k\Big),\label{eqn:kappa_cn}\\[2mm]
& \sin\frac{\theta}{2} = k\sn\Big(\sqrt{\mu}s; k\Big), \quad \mu=\frac{\lambda}{2k^2-1}.
\label{eqn:theta_cn}
\end{align}
\end{prop}
%
\begin{proof}
It can be easily verified that equations \eqref{eqn:kappa_dn} and \eqref{eqn:kappa_cn} satisfy
equation \eqref{eqn:elastica_ODE_kappa} from the differential equations for the $\dn$ and $\cn$
functions \cite{DLMF}
\begin{align}
& y=\dn(x,k),\quad \frac{d^2y}{dx^2} = (2-k^2)y - 2y^3,\\[2mm]
& y=\cn(x,k),\quad \frac{d^2y}{dx^2} = (2k^2-1)y - 2k^2y^3,
\end{align}
respectively, by applying suitable scale transformations. Also, equations \eqref{eqn:theta_dn} and
\eqref{eqn:theta_cn} are shown to satisfy equation \eqref{eqn:elastica_ODE_theta} in a similar manner.
\end{proof}
%
\begin{rem}
  The Jacobi elliptic functions can be extended to modules $k>1$, see, for example,
  \cite{lawden,DLMF}.  Thus, there exists an analytic continuation for all Jacobi elliptic functions
  in the range $k \geq 0$.  In this way, cases (i) and (ii) can be regarded as one.  As it is known,
  the case (i) (resp. (ii)) yields the elastica without (resp. with) inflection points; see Figure
  \ref{fig:smooth_discrete_elasticae}.
\end{rem}
%
\section{Integrable discrete Euler's elastica} \label{sec:disc.planar}
\subsection{Basic framework of discrete planar curves}
We first introduce the basic framework for discrete planar curves
\cite{hoffmann,IKMO:space_curve,Matsuura:mKdV}.  Let $\gamma_n \in \R[2]$ ($n\in\mathbb{Z}$) be a
discrete planar curve with $|\gamma_{n+1}-\gamma_n| = h$, where $h>0$ is a constant. We also assume
$\det{\big(\gamma_{n+1} - \gamma_n, \gamma_n - \gamma_{n-1}\big)} \neq 0$, i.e., not three
consecutive points are collinear.  Then $\gamma_n$ is called a discrete planar curve with segment length $h$.  We define the discrete tangent and normal vectors by
\begin{equation}
T_n = \frac{\gamma_{n + 1} - \gamma_n}{h}
= \left[\begin{array}{c}\cos\Theta_n \\\sin\Theta_n\end{array}\right]\quad\text{and}\quad
N_n = R\Big(\frac{\pi}{2}\Big)T_n,
\end{equation}
respectively, where the discrete angle function $\Theta_n$ is the angle of $T_n$ measured from the
horizontal axis in the counterclockwise direction.  The discrete Frenet frame is defined by
\begin{equation}
 \Phi_n = [T_n,N_n]\in{\rm SO}(2),
\end{equation}
and the discrete Frenet formula by
\begin{equation}\label{eqn:discrete_Frenet}
 \Phi_{n+1} = \Phi_nL_n,\quad
L_n = R(K_{n+1}),
\end{equation}
where $K_n = \Theta_{n}-\Theta_{n-1}$ is the angle between two adjacent tangent vectors.  Equation
\eqref{eqn:discrete_Frenet} is the discrete version of the Frenet formula
\eqref{eqn:Frenet_formula}; see Figure \ref{fig:discrete_Frenet}.  The {discrete curvature}
$\kappa_n$ can be defined as the reciprocal of the radius $\rho_n$ of the osculating circle touching
two adjacent segments at their midpoints \cite{hoffmann} (see Figure \ref{fig:discrete_curvature}),
\begin{equation}\label{eqn:discrete_curvature}
 \kappa_n = \frac{1}{\rho_n} = \frac{2}{h} \tan\frac{K_n}{2}.
\end{equation}
We note that the discrete Frenet formula \eqref{eqn:discrete_Frenet} can be written in terms of
$\kappa_n$ as
\begin{equation} \label{eqn:discrete_Frenet.kappa}
 \frac{\Phi_{n}-\Phi_{n-1}}{h} = \frac{\Phi_{n}+\Phi_{n-1}}{2}
\left[\begin{array}{cc} 0 &-\kappa_n \\ \kappa_n & 0\end{array}\right].
\end{equation}
%
\begin{figure}[h]
\centering
\begin{minipage}[b]{0.4\textwidth}
\centering
\includegraphics[scale=0.5]{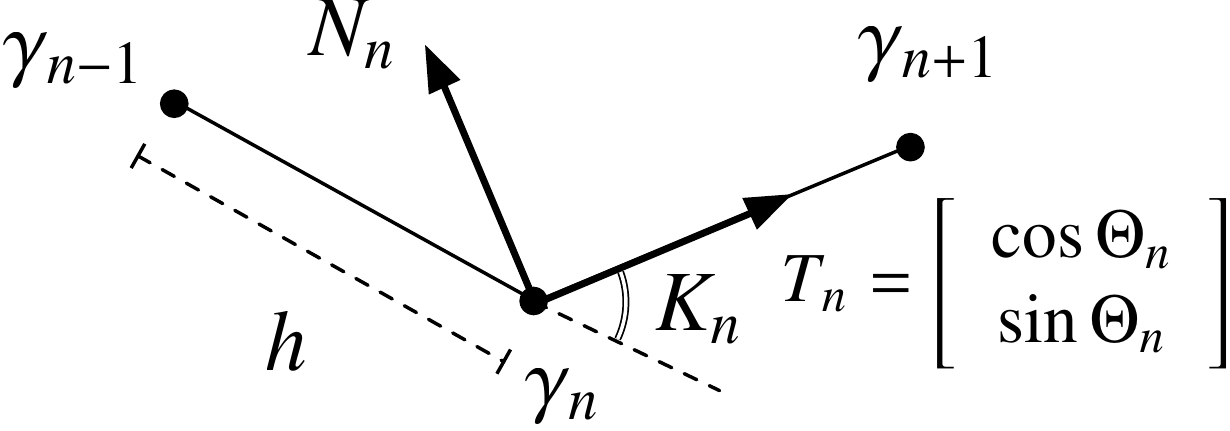}
\caption{Discrete planar curve and the Frenet frame.}  \label{fig:discrete_Frenet}  
\end{minipage}
\qquad
\begin{minipage}[b]{0.4\textwidth}
\centering
\includegraphics[scale=0.5]{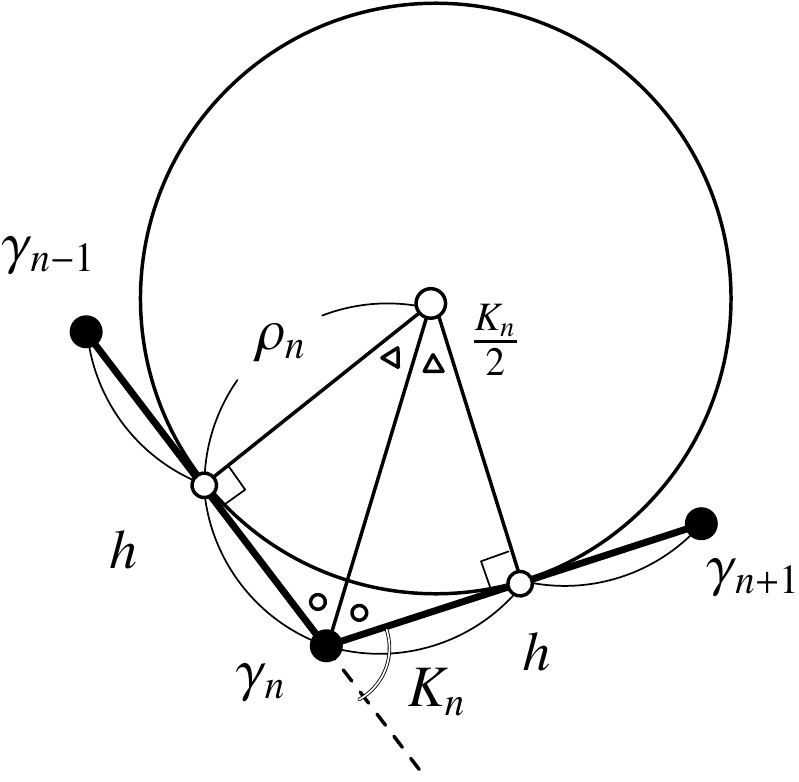}
\caption{Discrete curvature of discrete arc length planar curve.}  
\label{fig:discrete_curvature} 
\end{minipage}
\end{figure}
%
\subsection{Discrete Euler's elastica}
The discrete Euler's elastica \cite{bobenko.1999,Bobenko:LectureNote2007,Hoffmann_Kutz:2007} is
defined as a discrete planar curve $\gamma_n$ with segment length $h$ that is a critical point
of the functional
\begin{equation}
 E_d = \sum_{n=1}^{N-1} \frac{2}{h}\log\left(1+\frac{h^2}{4}\kappa_n^2\right)\cong
\sum_{n=1}^{N-1} \frac{2}{h}\log\big(1+\langle T_{n-1},T_n\rangle),
\end{equation}
with respect to variation with fixed endpoints and end edges. Note that $\cong$ means that the two
functionals yield the same critical points. As mentioned in \cite{Bobenko:LectureNote2007}, $E_d$
can be regarded as a discrete analogue of the elastic energy \eqref{eqn:elastic.energy}, in a sense
that $2/h\log(1+h^2\kappa_n^2/4)$ is the potential energy of the bending force proportional to the
discrete curvature $\kappa_n$ at each vertex.

Taking into account the preservation of $\norm{\gamma_{n} - \gamma_{n-1}}$ and $\gamma_N-\gamma_0$
by introducing the Lagrange multipliers $c_n\in\mathbb{R}$ and $a\in\mathbb{R}^2$, respectively,
consider the functional
\begin{equation}\label{eqn:Discrete_Energy_Functional2}
 S_d = \sum_{n=1}^{N-1} \frac{2}{h}\log\big(1+\langle T_{n-1},T_{n}\rangle)
+ \sum_{n=0}^{N-1}\big(c_{n}\langle T_{n},T_{n}\rangle + h\langle a,T_{n}\rangle\big).
\end{equation}
%
\begin{prop}\label{prop:Euler-Lagrange}
 The Euler-Lagrange equation for the functional
 \eqref{eqn:Discrete_Energy_Functional2} is given by
\begin{equation}\label{eqn:discrete_Euler-Lagrange}
 \frac{2}{h}\frac{T_{n-1}}{1+\langle T_{n-1},T_n\rangle}
 +\frac{2}{h}\frac{T_{n+1}}{1 + \langle T_{n},T_{n+1}\rangle} 
 +2c_nT_n + ha=0,
\end{equation}
or, equivalently, in terms of the discrete curvature $\kappa_n$ by
\begin{equation}\label{eqn:kappa}
 \kappa_{n+1}+\kappa_{n-1} = \frac{\alpha\kappa_n}{1+\frac{h^2}{4}\kappa_n^2},
\end{equation}
where $\alpha \in \R$ is a constant.
\end{prop}
%
\begin{proof}
  The variation of $S_d$ is calculated as follows:
\begin{align}
 \delta S_d &=
\sum_{n=1}^{N-1} 
\Big(\frac{2}{h}\frac{\langle \delta T_{n-1},T_n\rangle}{1+\langle T_{n-1},T_n\rangle}
+ \frac{2}{h}\frac{\langle T_{n-1},\delta T_n\rangle}{1+\langle T_{n-1},T_n\rangle}\Big)
+ \sum_{n=0}^{N-1}
\Big( 2c_n\langle T_n,\delta T_n\rangle + h\langle a,\delta T_n\rangle\Big)\nonumber\\
&= R + \sum_{n=1}^{N-2} \Big\langle
 \frac{2}{h}\frac{T_{n-1}}{1+\langle T_{n-1},T_n\rangle}
+\frac{2}{h}\frac{T_{n+1}}{1+\langle T_{n},T_{n+1}\rangle}
+2c_nT_n + ha,\delta T_n\Big\rangle,
\end{align}
where $R$ is the boundary term given by
\begin{displaymath}
R = \Big\langle \frac{2}{h}
\frac{T_{1}}{1+\langle T_{0},T_{1}\rangle}
+ 2c_0 T_0 + ha,\delta T_0\Big\rangle
+
\Big\langle\frac{2}{h}
\frac{T_{N-2}}{1+\langle T_{N-2},T_{N-1}\rangle}
+ 2c_{N-1}T_{N-1} + ha,\delta T_{N-1}\Big\rangle,
\end{displaymath}
which vanishes due to the boundary condition.  Setting $\delta S_d=0$ gives the Euler-Lagrange
equation,
\begin{equation}
 \frac{2}{h}\frac{T_{n-1}}{1+\langle T_{n-1},T_n\rangle}
+\frac{2}{h}\frac{T_{n+1}}{1+\langle T_{n},T_{n+1}\rangle}
+2c_nT_n + ha=0,
\end{equation}
which proves the first part. For the second part we use the discrete Frenet
formula \eqref{eqn:discrete_Frenet.kappa} written explicitly in terms of the tangent and normal
vectors,
\begin{equation}\label{eqn:discrete_Frenet2}
\frac{T_n - T_{n-1}}{h} = \kappa_n \frac{N_n + N_{n-1}}{2}, \qquad%
  \frac{N_n - N_{n-1}}{h} = -\kappa_n \frac{T_n + T_{n-1}}{2}. 
\end{equation}
From the discrete Frenet formula \eqref{eqn:discrete_Frenet} and equation
\eqref{eqn:discrete_curvature} we have that
\begin{equation}
 \frac{2}{h}\frac{\langle T_{n-1},N_n\rangle}{1+\langle T_{n-1},T_n\rangle} = -\kappa_n,\quad
 \frac{2}{h}\frac{\langle T_{n},N_n\rangle}{1+\langle T_{n},T_{n+1}\rangle} = \kappa_{n+1},
\end{equation}
thus taking the inner product of both hand sides of equation \eqref{eqn:discrete_Euler-Lagrange}
with $N_n$ gives
\begin{equation}\label{eqn:kappa_aN}
 \kappa_n - \kappa_{n+1}=h\langle a,N_n\rangle.
\end{equation}
Then, taking the inner product of the first equation in the discrete Frenet formula
\eqref{eqn:discrete_Frenet2} with $a$ gives
\begin{equation}
\braket{a}{T_n} - \braket{a}{T_{n-1}} = -\frac{1}{2} \kappa_n \kappa_{n+1} + \frac{1}{2} \kappa_{n-1} \kappa_{n}, 
\end{equation}
which implies that exists a constant $\lambda \in \R$ such that
\begin{equation}\label{eqn:Ta2}
\braket{a}{T_n} = -\frac{1}{2} \kappa_n \kappa_{n+1} + \lambda. 
\end{equation}
Finally taking the inner product of the second equation in the discrete Frenet formula
\eqref{eqn:discrete_Frenet2} with $a$, and using equations \eqref{eqn:kappa_aN} and \eqref{eqn:Ta2},
yields
\begin{equation}
 (\kappa_{n+1} + \kappa_{n-1})\left(1 + \frac{h^2}{4} \kappa_n^2\right) = (2 + h^2\lambda) \kappa_n,
\end{equation}
which is exactly equation \eqref{eqn:kappa} with $\alpha = 2 + h^2\lambda$.
\end{proof}
%
By using a technique similar to the one shown in \cite{takahashi}, we construct explicit solutions
to equation \eqref{eqn:kappa} from its discrete first integral,
\begin{equation}
 \kappa_{n+1}^2+\kappa_n^2 - \alpha\kappa_{n+1}\kappa_n + \frac{h^2}{4}\kappa_{n+1}^2\kappa_n^2=C,
\end{equation}
where $C \in \R$ is a constant (conserved quantity).  Here, we avoid the long computation required
and simply present the solutions corresponding to equations \eqref{eqn:kappa_dn} and
\eqref{eqn:kappa_cn}, and verify them by using the addition formulas for the Jacobi elliptic
functions.
%
\begin{prop}\label{eqn:discrete_solution_kappa}
Let $z\in\mathbb{R}$ be a constant. Then the following
functions satisfy equation \eqref{eqn:kappa}.\\ (i)
\begin{equation}\label{eqn:discrete_kappa_dn} 
\begin{split}
&\kappa_n = \frac{2}{h}\frac{\sn (k^{-1}z;k)}{\cn (k^{-1}z;k)}\,\dn\,(k^{-1}zn; k), \\[2mm]
& \alpha = 2\frac{\dn (k^{-1}z;k)}{\cn^2 (k^{-1}z;k)}.
\end{split}
\end{equation}
(ii) 
\begin{equation}\label{eqn:discrete_kappa_cn}
\begin{split}
&\kappa_n = \frac{2}{h}\frac{k\sn (z;k)}{\dn (z;k)}\,\cn(zn; k),\\[2mm]
& \alpha = 2\frac{\cn (z;k)}{\dn^2 (z;k)}.
\end{split}
\end{equation}
\end{prop}
%
\begin{proof}
  These solutions are verified directly by the addition formulas for the $\dn$ and $\cn$ functions
\cite{DLMF}.
On the one hand, we use that
\begin{equation} \label{eqn:dn.sum}
 \dn (u+v) + \dn (u-v)
=\frac{\frac{2\dn v}{\cn^2 v}\dn u}
{1+\frac{\sn^2v}{\cn^2v}\dn^2u}.
\end{equation}
Putting $u=k^{-1}zn$, $v=k^{-1}z$ and $\kappa_n = a\dn (u;k)$, and comparing equation
\eqref{eqn:dn.sum} with equation \eqref{eqn:kappa}, we see that
\begin{equation}
 \frac{1}{a^2}\frac{\sn^2v}{\cn^2v} = \frac{h^2}{4},\qquad
\frac{2\dn v}{\cn^2 v} = \alpha,
\end{equation}
which proves (i). Similarly, on the other hand we use
\begin{equation} \label{eqn:cn.sum}
 \cn(u+v) + \cn(u-v) 
=\frac{\frac{2\cn v}{\dn^2 v}\cn u}
{1+\frac{k^2\sn^2v}{\dn^2v}\cn^2u}.
\end{equation}
Then putting $u=zn$, $v=z$ and $\kappa_n = b\cn(u;k)$, we get
equation \eqref{eqn:kappa} with
\begin{equation}
\frac{1}{b^2} \frac{k^2\sn^2v}{\dn^2v}=\frac{h^2}{4},\qquad \frac{2\cn v}{\dn^2 v}=\alpha,
\end{equation}
which proves (ii).
\end{proof}
%
\begin{rem}\label{rem:eqn_kappa_n}\hfill
  \begin{enumerate}
  \item Comparing equations \eqref{eqn:kappa_dn} and \eqref{eqn:discrete_kappa_dn}, we see that
    there exists a constant $\Omega$ such that $\kappa_n = \kappa(\Omega n)$. Indeed, we have
    \begin{equation}
      \Omega = \frac{v}{\sqrt{\frac{\lambda}{2-k^2}}}\quad\text{and}\quad 
      v = \sn^{-1}\sqrt{\frac{\frac{\lambda h^2}{2-k^2}}{1+\frac{\lambda h^2}{2-k^2}}}.
    \end{equation}
    There is also a similar relationship between equations \eqref{eqn:kappa_cn} and
    \eqref{eqn:discrete_kappa_cn}.  This implies that the discrete curvature $\kappa_n$ is an
    ``exact discretization'' of the smooth curvature $\kappa(s)$.
  \item By putting $\alpha=h^2\lambda+2$ and $nh=s$, equation \eqref{eqn:kappa} yields equation
    \eqref{eqn:elastica_ODE_kappa} in the continuum limit $h\to 0$. On the level of solutions, the
    following parametrizations of $z$,
    \begin{numcases}
      {z=}
      \sqrt{\frac{\lambda}{2k^{-2}-1}}\,h & \mbox{{\rm for (i)}},\label{eqn:z.limit.i}\\[2mm]
      \sqrt{\frac{\lambda}{2k^2-1}}\,h & \mbox{{\rm for (ii)}},\label{eqn:z.limit.ii}
    \end{numcases} 
    are consistent in the continuum limit to equations \eqref{eqn:theta_dn} and
    \eqref{eqn:theta_cn}, respectively.
  \item Equation \eqref{eqn:kappa} is also known as the McMillan map, which is a special case of the
    Quispel-Roberts-Thompson (QRT) map solved by elliptic functions \cite{QRT}. It can also be
    regarded as an autonomous version of a discrete Painlev\'e II equation
    \cite{Kajiwara:dPII,KNY:topical_review,RGH}
  \item It is known that position vectors of both smooth and discrete elasticae admit explicit
    formulas in terms of the elliptic theta functions
    \cite{Matsuura:discrete_elastica_MSJ2020,Mumford}.
  \end{enumerate}
\end{rem}
Figure \ref{fig:smooth_discrete_elasticae} illustrates typical examples of both smooth and discrete elasticae.
\begin{figure}[h]
\centering
\begin{tabular}{ccc}
\includegraphics[width=4.5cm]{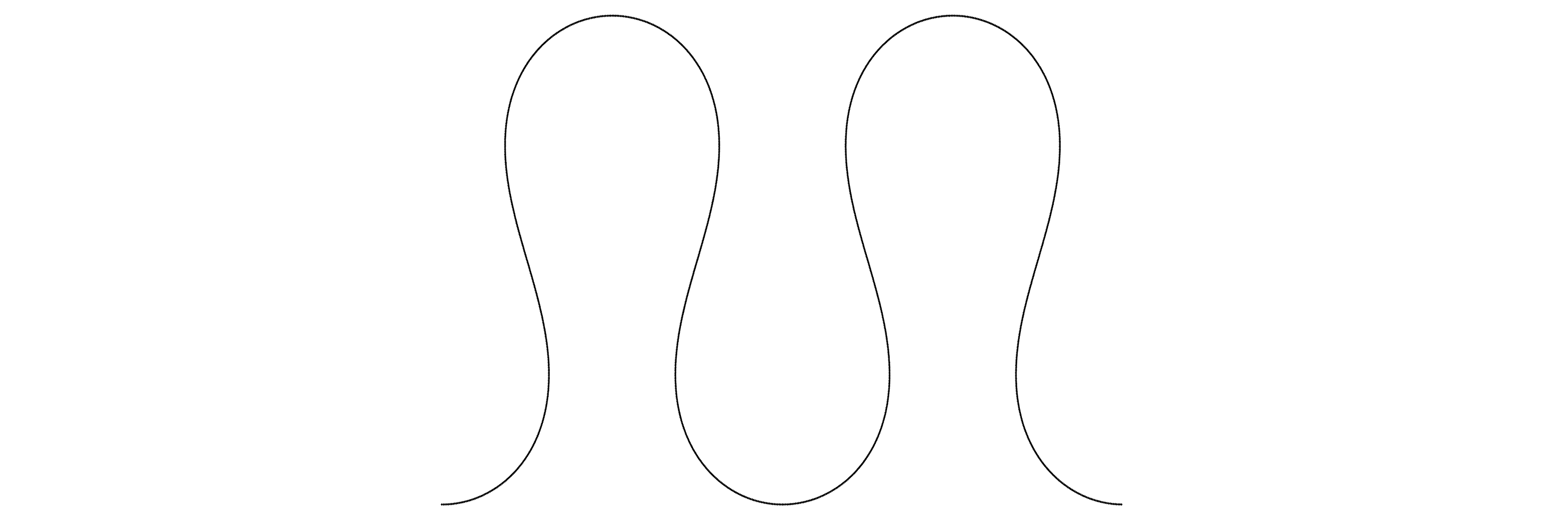} &\includegraphics[width=6cm]{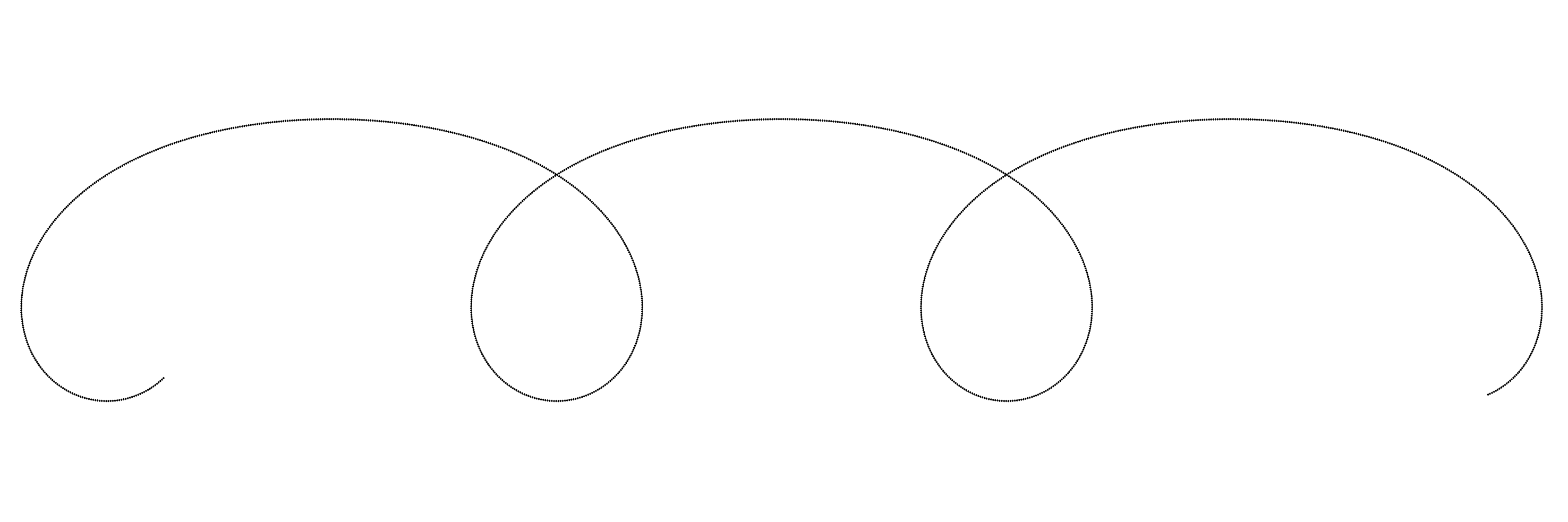} & \includegraphics[width=4.5cm]{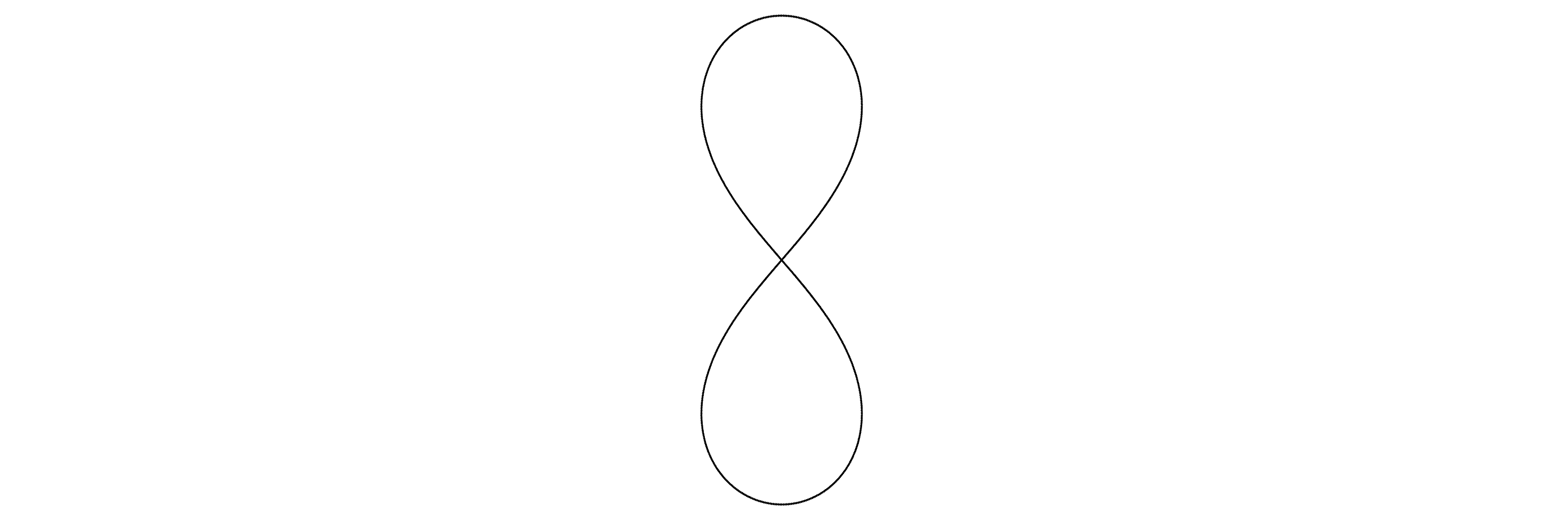} \\
\includegraphics[width=4.5cm]{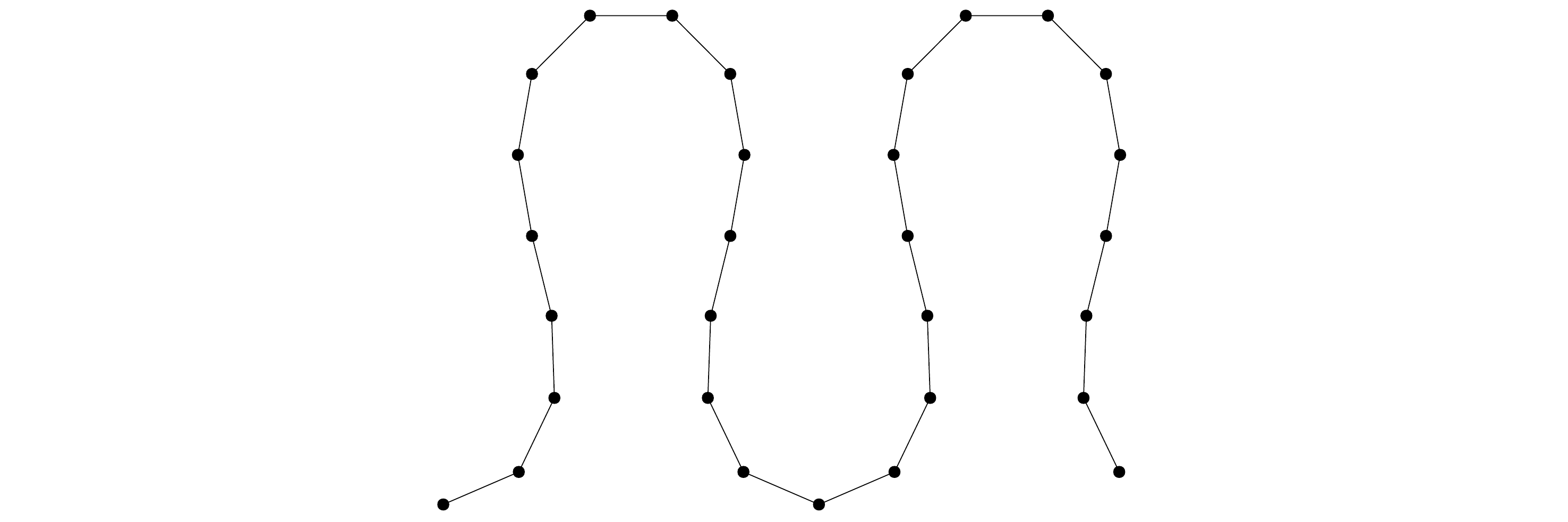} & \includegraphics[width=6cm]{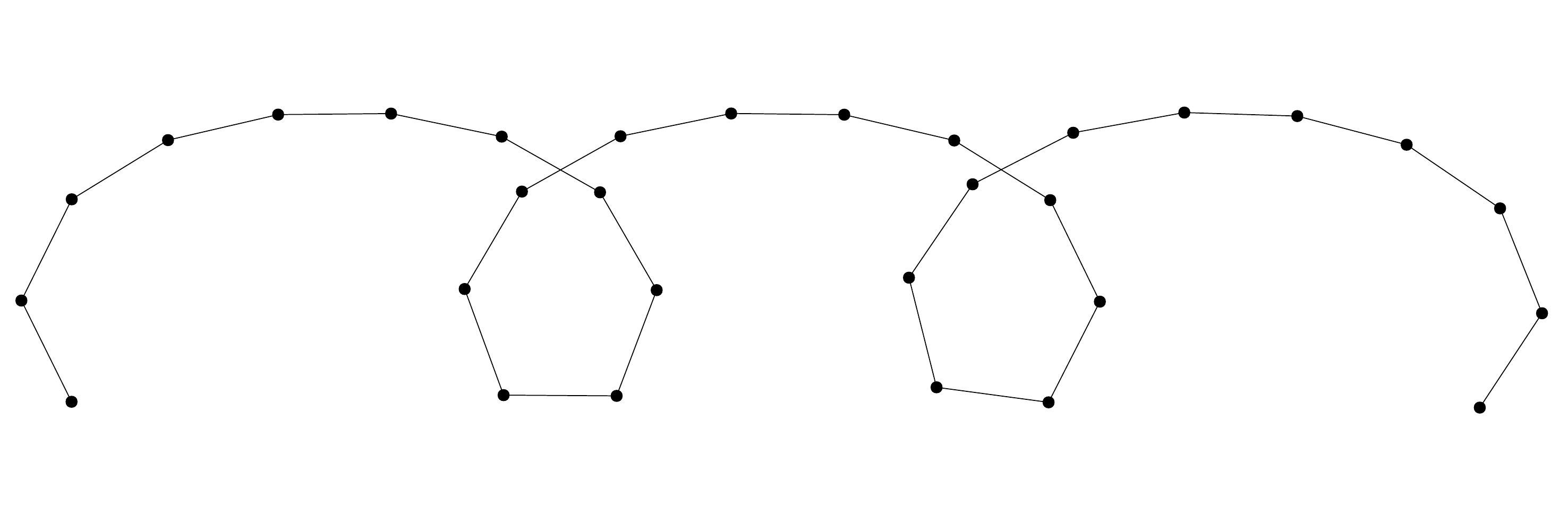} &\includegraphics[width=4.5cm]{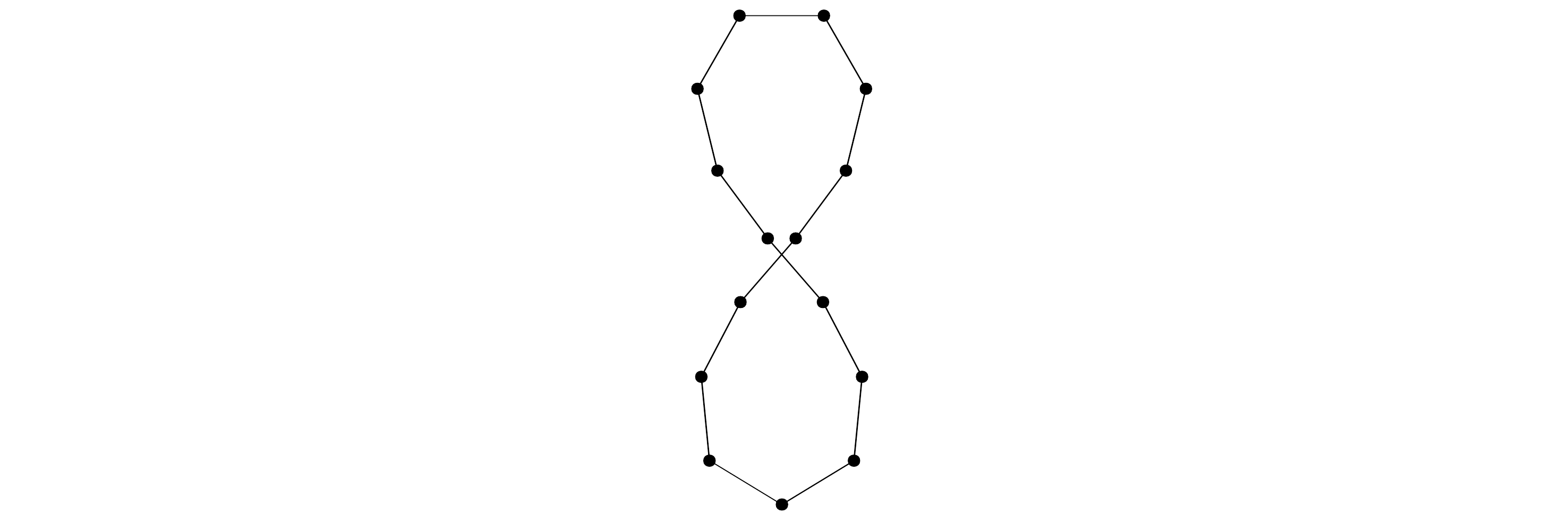} 
\end{tabular}
\caption{Typical examples of smooth and discrete elasticae. 
Left: (ii) $k=0.8$, middle: (i) $k=0.909$, right: (ii) $k=0.9089\ldots$ (smooth), 
$k=0.919\ldots$ (discrete).}  \label{fig:smooth_discrete_elasticae}  
\end{figure}
%
\subsection{Discrete Euler's elastica in terms of a potential function} \label{sec:disc.theta}
Following \cite{kajiwara}, we say that $\theta_n$ is a {\em potential function\/} if it is such that
\begin{equation}
 \Theta_n = \frac{\theta_{n+1}+\theta_n}{2}.
\end{equation}
In this context, the discrete curvature is written as
\begin{equation} \label{eqn:K&kappa_by_theta}
  \kappa_n = \frac{2}{h}\tan\left(\frac{\theta_{n+1}-\theta_{n-1}}{4}\right).
\end{equation}
%
\begin{prop}\label{prop:discrete_pendulum}
Suppose that $\theta_n$ satisfies
\begin{equation}\label{eqn:discrete_pendulum}
   \sin\left(\frac{\theta_{n+1} - 2\theta_n + \theta_{n-1}}{4}\right) 
+ \epsilon \sin\left(\frac{\theta_{n+1} + 2\theta_n + \theta_{n-1}}{4}\right) = 0,
\end{equation}
where $\epsilon\in\mathbb{R}$ is a constant. Then we have:
\begin{enumerate}
 \item  It holds that
\begin{equation}\label{eqn:conserved_quantity_theta}
\cos\left(\frac{\theta_{n+1}-\theta_{n}}{2}\right)
+ \epsilon\cos\left(\frac{\theta_{n+1}+\theta_{n}}{2}\right)
=\Lambda,
\end{equation}
where $\Lambda\in\mathbb{R}$ is a constant.
\item The discrete curvature satisfies equation \eqref{eqn:kappa} with
  $\alpha=2(1-\epsilon^2)/\Lambda^2$.
\end{enumerate}
\end{prop}
%
\begin{proof}
The first statement is shown as follows. Multiplying $\sin(\frac{\theta_{n+1}-\theta_{n-1}}{4})$ to
equation \eqref{eqn:discrete_pendulum}, using the product-to-sum formula, and rearranging terms gives:
\begin{equation}
 \cos\left(\frac{\theta_{n+1}-\theta_{n}}{2}\right)
+ \epsilon
\cos\left(\frac{\theta_{n+1}+\theta_{n}}{2}\right)
=
\cos\left(\frac{\theta_{n}-\theta_{n-1}}{2}\right)
+ \epsilon
\cos\left(\frac{\theta_{n}+\theta_{n-1}}{2}\right),
\end{equation}
which implies \eqref{eqn:conserved_quantity_theta}.  In order to show the second statement, we
introduce
\begin{equation}
\varphi_n = \frac{\theta_{n+1}-\theta_{n}}{2},\qquad
\psi_n = \frac{\theta_{n+1}+\theta_{n}}{2},
\end{equation}
for simplicity in the notation. We have that $K_n = \psi_{n}-\psi_{n-1} = \varphi_{n}+\varphi_{n-1}$,
and equation \eqref{eqn:discrete_pendulum} is rewritten as
\begin{equation}\label{eqn:dp1}
 \sin\left(\frac{\varphi_n-\varphi_{n-1}}{2}\right)
= - \epsilon\sin\left(\frac{\psi_n+\psi_{n-1}}{2}\right).
\end{equation}
We expand equation \eqref{eqn:dp1} as
\begin{equation}
  \sin\varphi_n\cos\frac{K_n}{2} - \cos\varphi_n\sin\frac{K_n}{2} =
  -\epsilon\sin\psi_n\cos\frac{K_n}{2} + \epsilon \cos\psi_n\sin\frac{K_n}{2},
\end{equation}
which gives
\begin{equation}\label{eqn:dp2}
  \tan\frac{K_n}{2}=\frac{\sin\varphi_n+\epsilon\sin\psi_n}{\Lambda}\quad\text{and}\quad
  \tan\frac{K_{n+1}}{2}=\frac{\sin\varphi_n-\epsilon\sin\psi_n}{\Lambda},
\end{equation}
where we used equation \eqref{eqn:conserved_quantity_theta}.
From equation \eqref{eqn:dp2} and the sum-to-product formula, we have
\begin{equation}\label{eqn:dp4}
  \Lambda\Big(\tan\frac{K_{n+1}}{2}+\tan\frac{K_{n-1}}{2}\Big) =2\sin\frac{K_n}{2}
  \Big[\cos\Big(\frac{\varphi_{n}-\varphi_{n-1}}{2}\Big)
  -\epsilon\cos\Big(\frac{\psi_{n}+\psi_{n-1}}{2}\Big)\Big].
\end{equation}
We use the following two expressions: On the one hand, from equation \eqref{eqn:discrete_pendulum},
\begin{equation} \label{eqn:aux01}
  \cos^2\Big(\frac{\varphi_{n}-\varphi_{n-1}}{2}\Big) - \epsilon^2\cos^2\Big(\frac{\psi_{n}+\psi_{n-1}}{2} \Big)
  = 1-\epsilon^2.
\end{equation}
On the other hand,
\begin{align} \label{eqn:aux02}
  \cos\Big(\frac{\varphi_{n}-\varphi_{n-1}}{2}\Big)+\epsilon\cos\Big(\frac{\psi_{n}+\psi_{n-1}}{2}\Big)
  &= \cos\frac{K_n}{2} \left( \cos\varphi_n + \epsilon \cos\psi_n \right) + \sin\frac{K_n}{2} \left(\sin\varphi_n + \epsilon \sin\psi_n \right) \nonumber\\
  &= \frac{\Lambda}{\cos\frac{K_n}{2}},
\end{align}
where we used equation \eqref{eqn:conserved_quantity_theta} and \eqref{eqn:dp2}.
Finally, we multiply equation \eqref{eqn:dp4} by equation \eqref{eqn:aux02} to obtain
\begin{equation}
  \Lambda\Big(\tan\frac{K_{n+1}}{2}
  + \tan\frac{K_{n-1}}{2}\Big){prop:sogo}
  \frac{\Lambda}{\cos\frac{K_n}{2}}
  = 2\sin\frac{K_n}{2}(1-\epsilon^2),
\end{equation}
where we used \eqref{eqn:aux01}, which is rewritten as
\begin{equation}\label{eqn:dp5}
 \tan\frac{K_{n+1}}{2}
+ \tan\frac{K_{n-1}}{2}
=
\frac{1-\epsilon^2}{\Lambda^2}
\frac{2\tan\frac{K_n}{2}}
{1+\tan^2\frac{K_n}{2}}.
\end{equation}
From the definition of discrete curvature, equation \eqref{eqn:dp5} is equivalent to equation
\eqref{eqn:kappa} with $\alpha=2(1-\epsilon^2)/\Lambda^2$, which proves the second statement.
\end{proof}
%
\begin{rem}
Proposition \ref{prop:discrete_pendulum} provides discrete analogues for equations
\eqref{eqn:elastica_ODE_theta} and \eqref{eqn:1st_integral_theta}.  In fact, by putting
\begin{equation} \label{eqn:lambda.limit}
  \Lambda = 1 - \frac{h^2}{4} \lambda,\quad \epsilon=\frac{h^2}{4}\mu,\quad s=nh,
\end{equation}
equations \eqref{eqn:discrete_pendulum} and \eqref{eqn:conserved_quantity_theta} yield equations
\eqref{eqn:elastica_ODE_theta} and \eqref{eqn:1st_integral_theta} in the continuum limit $h\to 0$.
\end{rem}
We next present explicit solutions for equation \eqref{eqn:discrete_pendulum}.  Part of Proposition \ref{prop:sol_discrete_pendulum} and \ref{prop:sogo} can be found in \cite{sogo}, in a slightly different context: in that work, the function $\theta_n$ is regarded as the angle function (here denoted as $\Theta_n$) instead of a potential function.
%
\begin{prop}\label{prop:sol_discrete_pendulum}
The following functions satisfy equation  \eqref{eqn:discrete_pendulum}:\\[2mm]
(i) 
\begin{equation}\label{eqn:sol_theta_1}
 \sin\frac{\theta_n}{2} = \sn(k^{-1}zn;k),\quad \dn(k^{-1}z;k)=\frac{1-\epsilon}{1+\epsilon}.
\end{equation}
(ii) 
\begin{equation}\label{eqn:sol_theta_2}
 \sin\frac{\theta_n}{2} = k\sn(zn;k),\quad \cn(z;k)=\frac{1-\epsilon}{1+\epsilon}.
\end{equation}
\end{prop}
%
\begin{proof}
For convenience, we write $u=k^{-1}zn$ and $v=k^{-1}z$.  For (i), we show that equation
 \eqref{eqn:sol_theta_1} satisfies equation
 \eqref{eqn:conserved_quantity_theta}.
 Note that $\cos(\theta_n/2) = \cn u$ and $\cos(\theta_{n+1}/2) = \cn (u + v)$.
 Then, the right hand side of equation \eqref{eqn:conserved_quantity_theta} is rewritten as
\begin{multline} \label{eqn:solt_1}
\cos\left(\frac{\theta_{n+1}-\theta_{n}}{2}\right)
+ \epsilon\cos\left(\frac{\theta_{n+1}+\theta_{n}}{2}\right) = \frac{1}{1-k^2\sn^2 u\sn^2 v} \\
\times \Big[\big\{(1+\epsilon)\cn^2 u + (1-\epsilon)\sn^2 u\dn v\big\}\cn v
+ \big\{-(1+\epsilon)\dn v+(1+\epsilon)\big\}\sn u\sn v\cn u\dn v\Big],
\end{multline}
where we used the addition formulas for the $\cn$ and $\sn$ functions.  Imposing
$\dn v = (1+\epsilon)/(1-\epsilon)$, we see that equation \eqref{eqn:conserved_quantity_theta} is
consistently reduced to $\Lambda=(1+\epsilon)\cn v$.  We prove (ii) in a similar manner.  In fact,
noticing that $\cos(\theta_n/2)=\dn u$ and imposing $\cn v = (1+\epsilon)(1-\epsilon)$, equation
\eqref{eqn:conserved_quantity_theta}, with $u=zn$ and $v=z$, is consistently reduced to
$\Lambda=(1+\epsilon)\dn v$.
\end{proof}
%
\begin{rem}\label{rem:theta_kappa_limit}\hfill
\begin{enumerate}
\item In case (i), the parameter $\alpha$ in equation \eqref{eqn:kappa} is given by
  \begin{equation}
    \alpha= 2\frac{(1-\epsilon^2)}{\Lambda^2} = 2 \frac{\dn (k^{-1}z;k)}{\cn^2 (k^{-1}z;k)},
  \end{equation}
  which implies that $\theta_n$ in equation \eqref{eqn:sol_theta_1} corresponds to $\kappa_n$ in
  equation \eqref{eqn:discrete_kappa_dn}.  In case (ii),
  \begin{equation}
    \alpha = 2 \frac{(1-\epsilon^2)}{\Lambda^2} = 2 \frac{\cn (z;k)}{\dn^2 (z;k)},
  \end{equation}
  so that $\theta_n$ in equation
  \eqref{eqn:sol_theta_2} corresponds to $\kappa_n$ in equation \eqref{eqn:discrete_kappa_cn}.
  These correspondences can be verified directly by computing
  $\kappa_n$ from equations
  \eqref{eqn:sol_theta_1} and \eqref{eqn:sol_theta_2}, respectively.
 \item Continuum limits of equations \eqref{eqn:sol_theta_1} and \eqref{eqn:sol_theta_2} 
to equations \eqref{eqn:discrete_kappa_dn} and \eqref{eqn:discrete_kappa_cn}, respectively, 
are obtained by putting $\epsilon=\frac{\mu}{4}h^2$ and taking the limit of $h\to 0$. This is 
consistent with Remark \ref{rem:eqn_kappa_n}.
\end{enumerate}
\end{rem}
%
We finally present the variational formulation for equation \eqref{eqn:discrete_pendulum}.
%
\begin{prop}[\cite{sogo}, Sec. 2] \label{prop:sogo}
  Equation \eqref{eqn:discrete_pendulum} is equivalent to the Euler-Lagrange equation of the functional
\begin{equation}
  \tilde{S}_d = \sum_{n=0}^{N-1} \cos\left(\frac{\theta_{n+1}-\theta_{n}}{2}\right)
  - \epsilon\cos\left(\frac{\theta_{n+1}+\theta_{n}}{2}\right),
\end{equation}
with respect to variations of the potential angle $\theta_n$ with fixed endpoints.
\end{prop}
%
\begin{proof}
  Let $L(\theta_n,\theta_{n+1})= \cos\left((\theta_{n+1}-\theta_{n})/2\right) -
\epsilon\cos\left((\theta_{n+1}+\theta_{n})/2\right)$. Then, the Euler-Lagrange equation is
calculated as
\begin{align}
 0 &= \frac{\partial}{\partial \theta_n} L(\theta_n,\theta_{n+1})
+ \frac{\partial}{\partial \theta_{n}} L(\theta_{n-1},\theta_{n}) \nonumber\\
&=\cos\left(\frac{\theta_{n+1}-\theta_{n-1}}{4}\right)\left[
\sin\left(\frac{\theta_{n+1}-2\theta_n+\theta_{n-1}}{4}\right)
+ \epsilon\sin\left(\frac{\theta_{n+1}+2\theta_n+\theta_{n-1}}{4}\right)\right],
\end{align}
which gives equation \eqref{eqn:discrete_pendulum}.
\end{proof}
%
\begin{rem}
  Equation \eqref{eqn:discrete_pendulum} can be seen as a reduction of two well known equations:
  \begin{enumerate}
    \item The discrete sine-Gordon equation \cite{Bobenko_Pinkall:discreteK_surface,Hirota:dsG,IKMO:space_curve},
      \begin{equation}\label{eqn:dsG}
        \sin\left(\frac{\theta_{l+1}^{m+1}-\theta_{l}^{m+1} - \theta_{l+1}^{m}+\theta_l^m}{4}\right)
        = \frac{a}{b}
        \sin\left(\frac{\theta_{l+1}^{m+1}+\theta_{l}^{m+1} + \theta_{l+1}^{m}+\theta_l^m}{4}\right),
      \end{equation}
      where $a$, $b$ are lattice intervals. In fact, assuming that $\theta$ depends only on $n=l+m$,
      equation \eqref{eqn:dsG} is reduced to equation \eqref{eqn:discrete_pendulum} with $\epsilon=-\frac{b}{a}$.
    \item The discrete potential modified KdV equation \cite{Hirota:dpmKdV},
      \begin{equation}
        \tan\frac{\theta_{l+1}^{m+1} - \theta_l^m}{4}
        = \frac{b+a}{b-a}
        \tan\frac{\theta_{l}^{m+1} - \theta_{l+1}^m}{4},
      \end{equation}
      or equivalently 
      \begin{equation}
        \sin\left(\frac{\theta_{l+1}^{m+1}-\theta_{l}^{m+1} + \theta_{l+1}^{m}-\theta_l^m}{4}\right)
        = \frac{a}{b}
        \sin\left(\frac{\theta_{l+1}^{m+1}+\theta_{l}^{m+1} - \theta_{l+1}^{m}-\theta_l^m}{4}\right),
      \end{equation}
      which describes the isoperimetric and equidistant deformation of discrete planar curves \cite{kajiwara,Matsuura:mKdV},
      is transformed to the discrete sine-Gordon equation \eqref{eqn:dsG} by $\theta_l^m\to (-1)^m\theta_l^m$.
      In this sense, equation \eqref{eqn:discrete_pendulum} can also be regarded as a reduction of
      the discrete potential modified KdV equation.
    \end{enumerate}
\end{rem}
%
\section{Approximation of discrete curves} \label{sec:approx.elastica}
In this section, we construct an algorithm to approximate a given discrete planar curve to a
discrete elastica.  Among the many possible discretizations for the elastica, the advantages of
using the one shown in this work can be described as follows: First, the discrete elasticae are
endowed with the same integrable structure as in their smooth counterpart, i.e., they possess
several conserved quantities, can be obtained via a variational principle, and their explicit
solutions are expressed in terms of Jacobi elliptic functions. Moreover, it is known that
variational integrators have controlled error in their solutions \cite{fernandez,
  marsden,patrick}. In particular, the explicit expression for the discrete curvature $\kappa_n$ is
an ``exact discretization'' of the smooth curvature $\kappa(s)$ as discussed in Remark
\ref{rem:eqn_kappa_n}, and the potential function $\theta_n$ has the same functional shape as the
smooth angle function $\theta(s)$.  From these observations, we expect this discretization to have
good numerical properties.
%
\subsection{General discrete Euler's elastica segment}
To describe a general curve segment in the plane, we include the freedom of rotation in the
equations.  We do this by shifting the angle function and discrete angle function by a constant
$\phi \in \R$ in all the expressions.  In particular, equation \eqref{eqn:discrete_pendulum} goes to
\begin{equation}\label{eqn:discrete_pendulum_phi}
   \sin\left(\frac{\theta_{n+1} - 2\theta_n + \theta_{n-1}}{4}\right) 
+ \frac{\mu h^2}{4} \sin\left(\frac{\theta_{n+1} + 2\theta_n + \theta_{n-1}}{4} - \phi\right) = 0,
\end{equation}
where we put $\epsilon = \mu h^2/4$, with $\mu > 0$ a constant.  From Proposition
\ref{prop:sol_discrete_pendulum}, we have 
\\[2mm]
\emph{(i)}
\begin{equation}\label{eqn:sol_theta_phi_i}
    \cos\frac{\theta_n - \phi}{2} = \cn(k^{-1}(zn + q);k),\quad
    \sin\frac{\theta_n - \phi}{2} = \sn(k^{-1}(zn + q);k),    
\end{equation}
\emph{(ii)}
\begin{equation}\label{eqn:sol_theta_phi}
    \cos\frac{\theta_n - \phi}{2} = \dn(zn + q;k),\quad
    \sin\frac{\theta_n - \phi}{2} = k \sn(zn + q;k),    
\end{equation}
where $k > 0$, $q, z \in \R$ are constants. The parameter $k$ determines the shape of the elastica,
$q$ the initial point, and $z$ is related with the length and point aggregation of the curve
segment. Finally, from its starting point $\gamma_0 \in \R[2]$, a discrete elastica segment is
calculated recursively by
\begin{equation} \label{eqn:reconstruct_gamma}
  \gamma_n = \gamma_{n-1} + h %
  \left[\begin{array}{c}\cos\left(\frac{\theta_{n+1} - \phi}{2} + \frac{\theta_{n} - \phi}{2} - \phi\right) \\
\sin\left(\frac{\theta_{n+1} - \phi}{2} + \frac{\theta_{n} - \phi}{2} - \phi\right)\end{array}\right]\qquad\text{for all } n = 1, \dots, N-1.
\end{equation}
Note that we can expand the sine and cosine in equation \eqref{eqn:reconstruct_gamma} and make use
of equation \eqref{eqn:sol_theta_phi} to obtain an explicit expression in terms of the Jacobi
elliptic functions.  We conclude that a general discrete elastica segment can be characterized by
seven parameters:
\begin{equation} \label{eqn:param}
  p = (x_0, y_0, h, \phi, z, q, k),
\end{equation}
where $x_0, y_0 \in \R$ are the two components of the initial point $\gamma_0$. We write as
$\gamma_n(p)$ to the discrete elastica with parameters $p$.
%
\subsection{Fairing process}
Given a general discrete curve segment $\zeta_n \in \R[2]$ ($n = 0, \dots, N-1$), we look for a
discrete elastica $\gamma_n(p) \in \R[2]$ that is the closest, in a $L^2$-distance sense, to
$\zeta_n$.  Namely, we seek to find $p^*$ such that
\begin{equation} \label{eq:opt}
  p^* = \argmin_{p \in U} \left\{\sum_{n = 0}^{N-1} \frac{1}{2}\norm{\gamma_n(p) - \zeta_n}^2 \right\},
\end{equation}
where
\begin{equation}
 U = \left\{(x_0, y_0, h, \phi, z, q, k) \st x_0, y_0, \phi, z, q \in \R \wedge h, k > 0\right\}.
\end{equation}
We solve this non-convex problem using the Interior Point Optimizer (IPOPT) package, that for our
purpose can be seen as a gradient descent-like method for nonlinear optimizations \cite{ipopt}. For
its implementation we need to compute the gradient and the Hessian of the objective function
\begin{equation}
 \mathcal{L}(p) := \sum_{n = 0}^{N-1} \frac{1}{2}\norm{\gamma_n(p) - \zeta_n}^2.
\end{equation}
For the gradient, we have
\begin{equation}
 \pderiv{}{p_i}\mathcal{L}(p) = \sum_{n=0}^{N-1} \braket{\gamma_n(p) - \zeta_n}{\pderiv{}{p_i}\gamma_n(p)}, %
  \qquad p_i = x_0, y_0, h, \phi, z, q, k,
\end{equation}
which is computed recursively from equation \eqref{eqn:reconstruct_gamma},
\begin{equation}
 \pderiv{}{p_i}\gamma_n(p) = \pderiv{}{p_i}\gamma_{n-1}(p) +
  \begin{cases}
    0 & p_i = x_0, y_0,\\
    T_n & p_i = h,\\
    h \pderiv{}{p_i}T_n & \text{otherwise},
  \end{cases}\qquad \pderiv{}{p_i}\gamma_0 =
  \begin{cases}
    \left(\begin{smallmatrix}1\\0\end{smallmatrix}\right) & p_i = x_0,\\
    \left(\begin{smallmatrix}0\\1\end{smallmatrix}\right) & p_i = y_0,\\
    \left(\begin{smallmatrix}0\\0\end{smallmatrix}\right) & \text{otherwise},
  \end{cases}
\end{equation}
Then, the only non-trivial derivatives we need to compute are
\begin{equation}
  \pderiv{}{{p}_i} T_n = N_n \left(\frac{1}{C_{n+1}(p)} \pderiv{}{{p}_i}S_{n+1}(p) +
    \frac{1}{C_n(p)}\pderiv{}{{p}_i}S_n(p)\right),\qquad {p}_i = z, q, k,
\end{equation}
or equivalently,
\begin{equation}
  \pderiv{}{{p}_i} T_n = - N_n \left(\frac{1}{S_{n+1}(p)} \pderiv{}{{p}_i}C_{n+1}(p) +
    \frac{1}{S_n(p)}\pderiv{}{{p}_i}C_n(p)\right),\qquad {p}_i = z, q, k,
\end{equation}
where we denoted $S_n \equiv \sin\frac{\theta_n - \phi}{2}$ and
$C_n \equiv \cos\frac{\theta_n - \phi}{2}$.  Finally, we use equations \eqref{eqn:sol_theta_phi_i},
\eqref{eqn:sol_theta_phi} and the derivatives of the Jacobi elliptic functions with respect to their
argument and module to complete the computation.  For the Hessian, we use a numerical quasi-Newton
approximation, internally computed by the package.

The IPOPT method needs a starting point $\hat{p}$, that we refer as the {\em initial guess}.  In the
following subsection we describe the algorithm that we use to obtain the initial guess, which is a
discrete analogue of the one provided in \cite{brander}.
%
\subsection{Initial parameters} \label{sec:affine.kappa}

The initial guess, that starts the IPOPT method, can be obtained in a numerically stable manner
thanks to two geometric properties of the discrete elastica: Proposition \ref{prop:guess} and
Corollary \ref{prop:guess}. Remarkably, these are geometrically equivalent to the same properties
for the smooth elastica \cite{brander}.  Let
\begin{equation}
\mathrm{I}= \left[\begin{array}{c} \sin\phi\\ -\cos\phi\end{array}\right] \in \R[2],
\end{equation}
and define the projection of the curve $\gamma_n$ onto $\mathrm{I}$ as
\begin{equation}\label{eqn:def_u_n}
u_n = \braket{\mathrm{I}}{\gamma_n},
\end{equation}
and the angle measured from $\mathrm{I}$ as $\Psi_n = \frac{\pi}{2} + \Theta_n -\phi$ or,
equivalently, such that
\begin{equation}\label{eqn:def_Psi_n}
\left[\begin{array}{c} \cos\Psi_n\\ \sin\Psi_n\end{array}\right] = R\left(\frac{\pi}{2} - \phi\right) T_n. 
\end{equation}
%
\begin{prop} \label{prop:guess}
  The discrete curvature $\kappa_n$ is an affine function of the projection $u_n$, satisfying
  \begin{equation} \label{eqn:affine.kappa}
    \kappa_n = \frac{\mu}{\Lambda} u_n + A ,
  \end{equation}
  where $\Lambda \in \R$ satisfies equation \eqref{eqn:conserved_quantity_theta} and $A \in \R$ is a
  constant.
\end{prop}
%
\begin{proof}
  In the context of the proof of Proposition \ref{prop:discrete_pendulum}, after incorporating
  $\phi$ and putting $\epsilon = \mu h^2/4$, from equation \eqref{eqn:dp2} we obtain
\begin{equation}\label{eqn:kappa_Psi}
 \frac{\kappa_{n+1} - \kappa_n}{h} = -\frac{\mu}{\Lambda}\sin(\Psi_n - \phi).
\end{equation}
Then, noticing that $u_{n+1} - u_n = h \braket{I}{T_n}$, we have
\begin{equation}\label{eqn:u_Psi}
\frac{u_{n+1} - u_n}{h} = -\sin(\Psi_n - \phi). 
\end{equation}
Hence, by comparing equations \eqref{eqn:affine.kappa} and \eqref{eqn:kappa_Psi}, we conclude that
there exists a constant $A \in \R$ such that, for all $n$,
\begin{equation}
  \kappa_n = \frac{\mu}{\Lambda} u_n + A.
\end{equation}
\end{proof}
Note that, by putting $\mu_1 = \mu \cos\phi$ and $\mu_2 = \mu\sin\phi$, equation
\eqref{eqn:affine.kappa} can be expressed as
\begin{equation} \label{eqn:affine.kappa.aux}
\kappa_n = \frac{1}{\Lambda} (\mu_2 x_n - \mu_1 y_n) + A,
\end{equation}
where $x_n, y_n \in \R$ are the two components of $\gamma_n$.
%
\begin{cor} \label{cor:guess}
It holds that
\begin{equation} \label{eqn:quad.sin}
  \sin\Psi_n = \frac{\mu}{2 \Lambda} u_{n+1} u_n +  A \frac{u_{n+1} + u_n}{2} + B,
\end{equation}
where $B \in \R$ is a constant.
\end{cor}
%
\begin{proof}
From the definition of $u_n$ and $\Psi_n$ in equations \eqref{eqn:def_u_n} and \eqref{eqn:def_Psi_n}, respectively,
we obtain
\begin{equation}
  T_n =  R(\phi) \left[\begin{array}{c} \sin\Psi_n\\[2mm] -\frac{u_{n+1} - u_n}{h}\end{array}\right], \quad %
  N_n = R(\phi) \left[\begin{array}{c} \frac{u_{n+1} - u_n}{h} \\[2mm] \sin\Psi_n\end{array}\right].
\end{equation}
Then, putting this into the discrete Frenet formula $(T_n - T_{n-1})/h = \kappa_n (N_n + N_{n-1})/2$ gives
\begin{equation}
\sin\Psi_n - \sin\Psi_{n-1} = \frac{1}{2}\left(\frac{\mu}{\Lambda} u_n + A\right) \left(u_{n+1} - u_{n-1}\right), 
\end{equation}
where we used \eqref{eqn:affine.kappa}.  After expanding the right hand side of the previous
expression and then adding $\pm \frac{A}{2} u_n$, we conclude that there exists a constant
$B \in \R$ such that, for all $n$,
\begin{equation}
  \sin\Psi_n = \frac{\mu}{2 \Lambda} u_{n+1} u_n +  A \frac{u_{n+1} + u_n}{2} + B.
\end{equation}
\end{proof}
%
To estimate $(\phi, z, q, k)$ we use some results from the smooth elastica to avoid unnecessary
complexity in the discrete case. We use the following approximations: From equation
\eqref{eqn:conserved_quantity_theta} with $\epsilon = \mu h^2/4$, taking equation
\eqref{eqn:lambda.limit} into account, equations \eqref{eqn:affine.kappa.aux} and
\eqref{eqn:quad.sin} can be expanded in terms of $\sqrt{\mu}h$ as
\begin{equation} \label{eq:affine.order2}
  \kappa_n = \mu_2 x_n - \mu_1 y_n + A + \ord{\mu h^2},
\end{equation}
and
\begin{equation} \label{eqn:quad.sin.order2}
   \sin\Psi_n = \frac{1}{2} \mu u_n^2 + A u_n + B + \ord{\mu h^2},
 \end{equation}
respectively. For the discrete curvature, note that solutions \eqref{eqn:discrete_kappa_dn} and
\eqref{eqn:discrete_kappa_cn} can be written respectively as
\begin{equation} \label{eqn:kappa.withmax}
\left\{
\begin{array}{lll}
\text{(i)}&\kappa_n =  \kappa_{\mathrm{max}} \dn(k^{-1}(z n + q); k),
&\kappa_{\mathrm{max}} = \dfrac{2}{h}\, \dfrac{\sn(k^{-1}z;k)}{\cn(k^{-1}z;k)},\\[6mm]
\text{(ii)}&\kappa_n = \kappa_{\mathrm{max}}  \cn(z n + q; k),
& \kappa_{\mathrm{max}} = \dfrac{2}{h}\, \dfrac{k \sn(z;k)}{\dn(z;k)}.
\end{array}
\right.
\end{equation}
From Remark \ref{rem:theta_kappa_limit} we have $\alpha = 2 + \lambda h^2 + \order[4]$, and it
follows from equations \eqref{eqn:z.limit.i} and \eqref{eqn:z.limit.ii} that $z = \sqrt{\mu} h + \order[2]$.
Then we can approximate $\kappa_{\mathrm{max}}$ as
\begin{equation} \label{eqn:kappa.order2}
  \begin{cases}
    \text{(i)}& \kappa_{\mathrm{max}} =  2k^{-1}\sqrt{\mu} + \ord{\mu h^2},\\[4mm]
    \text{(ii)}& \kappa_{\mathrm{max}} =  2k\sqrt{\mu} + \ord{\mu h^2}.
  \end{cases}
\end{equation}
We obtain an approximation of the parameter $k$ from these expressions.  We see from
equation \eqref{eqn:quad.sin.order2} that $u$ must be bound from above and below, the upper bound
$u_{\mathrm{max}}$ being
\begin{equation} \label{eqn:umax}
  u_{\mathrm{max}} = \frac{-A + \Delta}{\mu} + \ord{h^2},\quad
\Delta = \sqrt{A^2 - 2 \mu (B - 1)}.
\end{equation}
Noticing that $u_{\textrm{max}}$ occurs at the same instance as $\kappa_{\textrm{max}}$, from
equation \eqref{eqn:affine.kappa} we have
\begin{equation} \label{eqn:kappamax}
  \kappa_{\mathrm{max}} = \Delta + \ord{\mu h^2}. 
\end{equation}
Hence, from equation \eqref{eqn:kappa.order2} we obtain
\begin{equation} \label{eqn:k.order2}
  \begin{cases}
    \text{(i)}& k = \dfrac{2\sqrt{\mu}}{\Delta} + \ord{\sqrt{\mu} h^2},\\[4mm]
    \text{(ii)}&k = \dfrac{\Delta}{2\sqrt{\mu}}  + \ord{\sqrt{\mu} h^2}.
  \end{cases}
\end{equation}
%
\noindent{\bf Pseudo code:}
Given a discrete curve $\gamma_n$ ($n = 0,1,\dots, N-1$) with segment length $h$, from their
definition, we compute $\Theta_n$ such that
\begin{equation}
  \Theta_n : \left[\begin{array}{c} \cos\Theta_n \\ \sin\Theta_n\end{array}\right] 
  = \frac{\gamma_{n+1} - \gamma_n}{h},\qquad\text{for all } n = 0, 1, \dots, N-2,
\end{equation}
and the discrete curvature $\kappa_n$ by
\begin{equation}
  \kappa_n := \frac{2}{h} \tan\frac{\Theta_n - \Theta_{n-1}}{2},\qquad\text{for all } n = 1, 2, \dots, N-2.
\end{equation}
Then, we obtain the initial guess $\hat{p}$ by solving equations \eqref{eq:affine.order2} and
\eqref{eqn:quad.sin.order2} in the least-square sense and using several of the equations mentioned
above.  We proceed as follows:

\begin{itemize}
\item (parameter $\phi$)
From equation \eqref{eq:affine.order2}, compute
\begin{equation}
 (\hat\mu_1, \hat\mu_2, \hat{A}) = \argmin_{(\mu_1,\mu_2,A)}\left\{\sum_{n=1}^{N-2}
\left(\kappa_n + \mu_1 y_n - \mu_2 x_n - A\right)^2\right\},
\end{equation}
then $\hat\mu = \sqrt{\hat\mu^2_1 + \hat\mu^2_2},$ and $\phi$ is such that $\cos\phi = \hat\mu_1/\hat\mu$ and $\sin\phi = \hat\mu_2/\hat\mu$.
\item (parameter $k$)   From equation \eqref{eqn:quad.sin.order2}, compute
\begin{equation}
\hat{B} = \argmin_{B}\left\{\sum_{n=0}^{N-2}\left(\sin\Psi_n - \frac{1}{2} \hat{\mu} u_n^2 - \hat{A} u_n - B\right)^2\right\}.
\end{equation}
From equation \eqref{eqn:k.order2}, if $\hat{B} < \frac{\hat{A}^2}{2\hat{\mu}} - 1$ we are in case (i) and
\begin{equation}
  \hat{k} = 2 \left(\frac{\hat{A}^2}{\hat{\mu}} - 2 (\hat{B} - 1)\right)^{-1/2},
\end{equation}
otherwise we are in case (ii) and
\begin{equation}
 \hat{k} = \frac{1}{2}\left(\frac{\hat{A}^2}{\hat{\mu}} - 2 (\hat{B} - 1)\right)^{1/2}.
\end{equation}
\item (parameter $q$ and $z$) For simplicity, let $s_n = z n + q$. Define $m \in \N$ as the number
of segments in which the function $u_n$ is monotone. We counted $m$ manually, although it
could also be estimated by, for example,
\begin{equation}
\begin{cases}\bigskip
\text{(i):}\quad \hat{m} = \left\lceil(N-1)h\,\dfrac{\sqrt{\hat{\mu}}}{K(\hat{k})}\right\rceil \,({}+1),\\
\text{(ii):}\quad \hat{m} = \left\lceil(N-1)h\,\dfrac{\sqrt{\hat{\mu}}}{2K(\hat{k})}\right\rceil \,({}+1),
\end{cases} 
\end{equation}
where $K$ is the complete elliptic integral of the first kind, and the term in brackets $(+1)$ is added
only if both $u_0$ and $u_{N-1}$ are simultaneously increasing or decreasing.
Now we can simply invert the Jacobi elliptic function at the endpoints $n=0$ and $n=N-1$ to obtain $q$ and $z$.
From equations \eqref{eqn:kappa.withmax}  and \eqref{eqn:kappa.order2}, we have the following:
\begin{itemize}
\item[(i)]
  \begin{equation}
    \dn(\hat{k}^{-1}s_n; \hat{k}) = \frac{\hat{\mu} u_n + \hat{A}}{2\hat{k}^{-1}\sqrt{\hat{\mu}}} 
  \end{equation}
  which can be rewritten as
  \begin{equation}
    \sn(\hat{k}^{-1}s_n; \hat{k}) = \hat{k}^{-1} 
    \sqrt{ 1 - \left(\frac{\hat{\mu} u_n + \hat{A}}{2\hat{k}^{-1}\sqrt{\hat{\mu}}}\right)^2} \equiv U_n.
  \end{equation}
  Hence,
  \begin{itemize}
  \item[--] If $u_n$ is decreasing on the first segment:
    \begin{equation} s_0 = \hat{k} F(\arcsin U_0; \hat{k}), \end{equation}
    and
    \begin{equation} s_{N-1} = (m-1) \hat{k} K(\hat{k}) + \hat{k} F(\arcsin U_{N-1}; \hat{k}), \end{equation}
    if $m$ is odd, or
    \begin{equation} s_{N-1} = m \hat{k} K(\hat{k}) - \hat{k} F(\arcsin U_{N-1}; \hat{k}), \end{equation}
    if $m$ is even.
  \item[--] If $u_n$ is increasing on the first segment:
    \begin{equation} s_0 = 2 \hat{k} K(\hat{k}) - \hat{k} F(\arcsin U_0; \hat{k}), \end{equation}
    and
    \begin{equation} s_{N-1} = (m+1) \hat{k} K(\hat{k}) - \hat{k} F(\arcsin U_{N-1}; \hat{k}), \end{equation}
    if $m$ is odd, or
    \begin{equation} s_{N-1} = m \hat{k} K(\hat{k}) + \hat{k} F(\arcsin U_{N-1}; \hat{k}), \end{equation}
    if $m$ is even.
  \end{itemize}
  
\item[(ii)]
  \begin{equation} \cn(zn + q; \hat{k}) = \frac{\hat{\mu} u_n + \hat{A}}{2\hat{k}\sqrt{\hat{\mu}}} \equiv U_n. \end{equation}
  Hence,
  \begin{itemize}
  \item[--] If $u_n$ is decreasing on the first segment:
    \begin{equation} s_0 = F(\arccos U_{0}; \hat{k}), \end{equation}
    and
    \begin{equation} s_{N-1} = 2(m-1) K(\hat{k}) + F(\arccos U_{N-1}; \hat{k}), \end{equation}
    if $m$ is odd, or
    \begin{equation} s_{N-1} = 2m K(\hat{k}) - F(\arccos U_{N-1}; \hat{k}), \end{equation}
    if $m$ is even.
  \item[--] If $u_n$ is increasing on the first segment:
    \begin{equation} s_0 = 4K(\hat{k}) - F(\arccos U_0; \hat{k}), \end{equation}
    and
    \begin{equation} s_{N-1} = 2(m+1)K(\hat{k}) - F(\arccos U_{N-1}; \hat{k}), \end{equation}
    if $m$ is odd, or
    \begin{equation} s_{N-1} = 2m K(\hat{k}) + F(\arccos U_{N-1}; \hat{k}), \end{equation}
    if $m$ is even.
  \end{itemize}
\end{itemize}
Here we used the fact that $\sn^{-1} = F \circ \arcsin$, $\cn^{-1} = F \circ \arccos$, 
and $F$ is the elliptic integral of the first kind. Finally,
\begin{equation}
\hat{q} = s_0, \qquad \hat{z} = \frac{1}{N-1} (s_{N-1} - s_0). 
\end{equation}
\item (parameters $x_0$ and $y_0$) From the previous steps, using all the recovered parameters,
construct a discrete elastica segment that starts at the origin, $\hat{\gamma}$. Then,
\begin{equation}
(\hat{x}_0, \hat{y}_0) = \argmin_{(x_0, y_0)}\left\{\sum_{n=0}^{N-1}\left(\gamma_n - \hat{\gamma}_n\right)^2\right\}. 
\end{equation}
\end{itemize}
%
Figure \ref{fig:fairing} illustrates typical examples of the fairing by the discrete elasticae
obtained by using the above algorithm.
\begin{figure}[h]
\centering
\includegraphics[width=6cm]{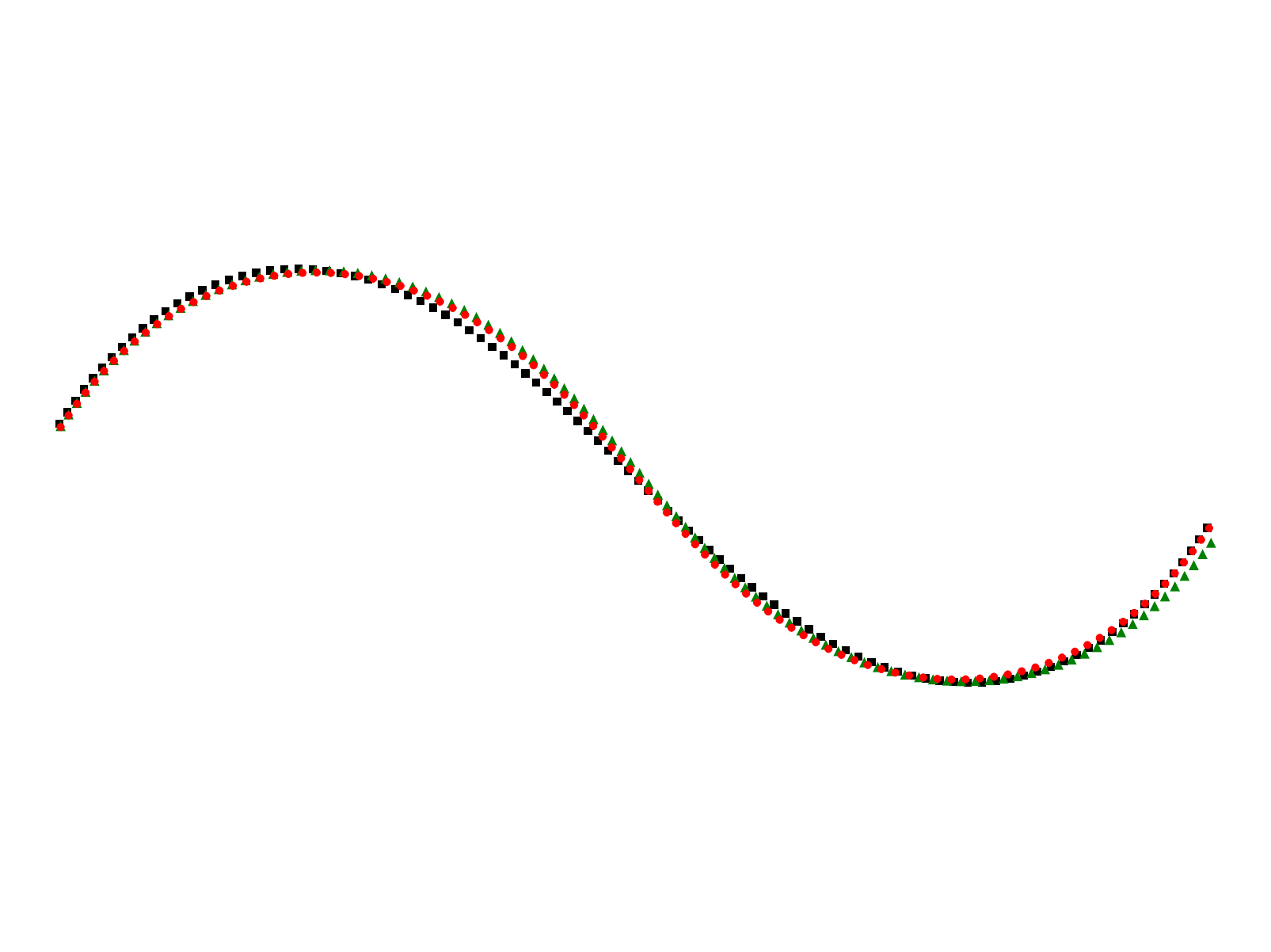}\includegraphics[width=6cm]{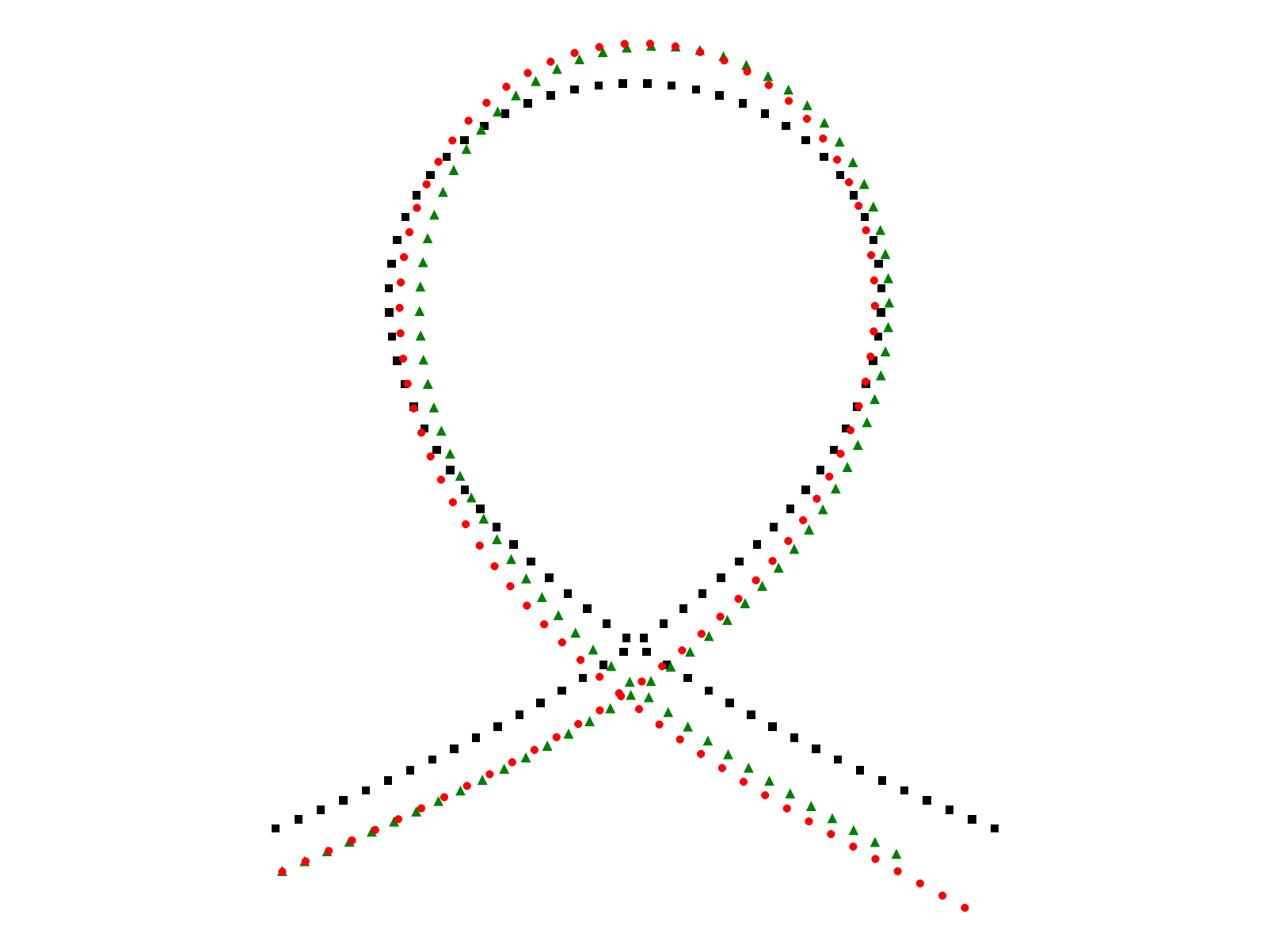}
\includegraphics[width=6cm]{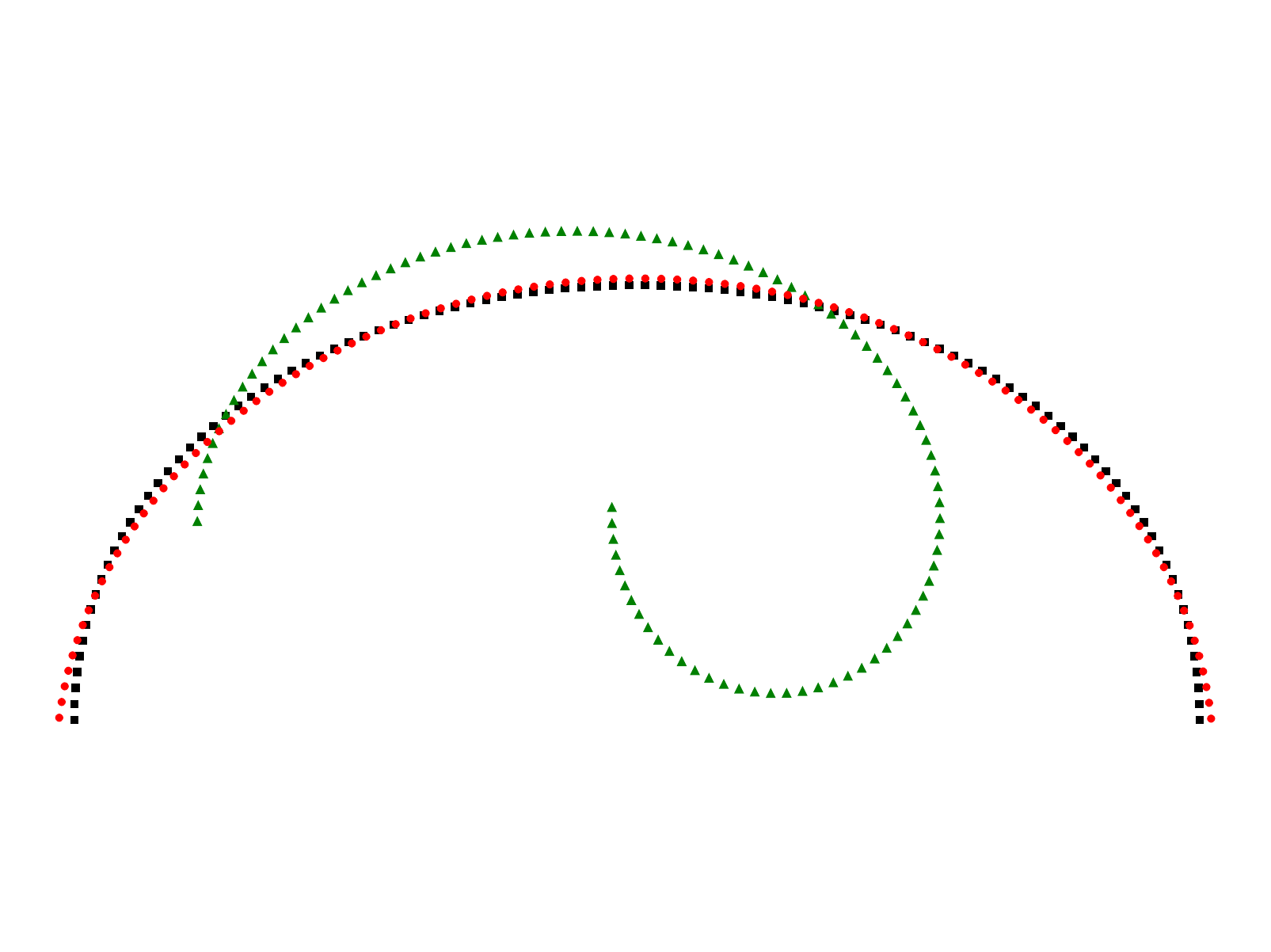} \\
\caption{Typical examples of the fairing by the discrete elasticae.  Black squares: input curve,
green triangles: initial guess, red circles: output elastica.}  \label{fig:fairing}
\end{figure}
%
\section{Application: characterization of keylines of Japanese handmade pantiles}\label{sec:application}
%
\subsection{Background and outline}\label{subsec:background}
\emph{Sangawara} (Japanese pantiles) are the most common type of roof tiles in Japan and are thought
to be unique to the country.  Although the number of buildings with sangawara roofs are decreasing,
the landscapes with the remaining constructions having sangawara roofs are consider by the community
to be one of the most beautiful scenes and culturally Japanese.  Traditionally, sangawara were
handmade from local clay by placing a clay plate on a wooden mold, beating it with a board called
\emph{tataki}, and stroking it with a board called \emph{nadeita} (Figure \ref{fig:sec5.fig6}). In
recent times, they are mass-produced by metal mold presses in limited areas. The mold shapes are
thought to be based on the shape of the sangawara in the handmade era, but companies keep their
designs a trade secret, so it is not clear.  We consider that it is important to characterize
aesthetically pleasing curves like sangawara with mathematical formulas to be used in architectural
design.  Because of this, and the fact that the process involves bending the clay plate, we thought
that the shape of sangawara could possibly be approximated by elasticae.  
In Section \ref{sec:5.2} we explain how the handmade sangawara (simply referred to as pantiles) were
collected, in Section \ref{sec:5.3} we obtain the keyline of each pantile, and in Section
\ref{sec:5.4} we approximate those keylines to discrete elasticae.
%
\begin{figure}[h]
\centering
\begin{tabular}{ccc} 
\includegraphics[width=6cm]{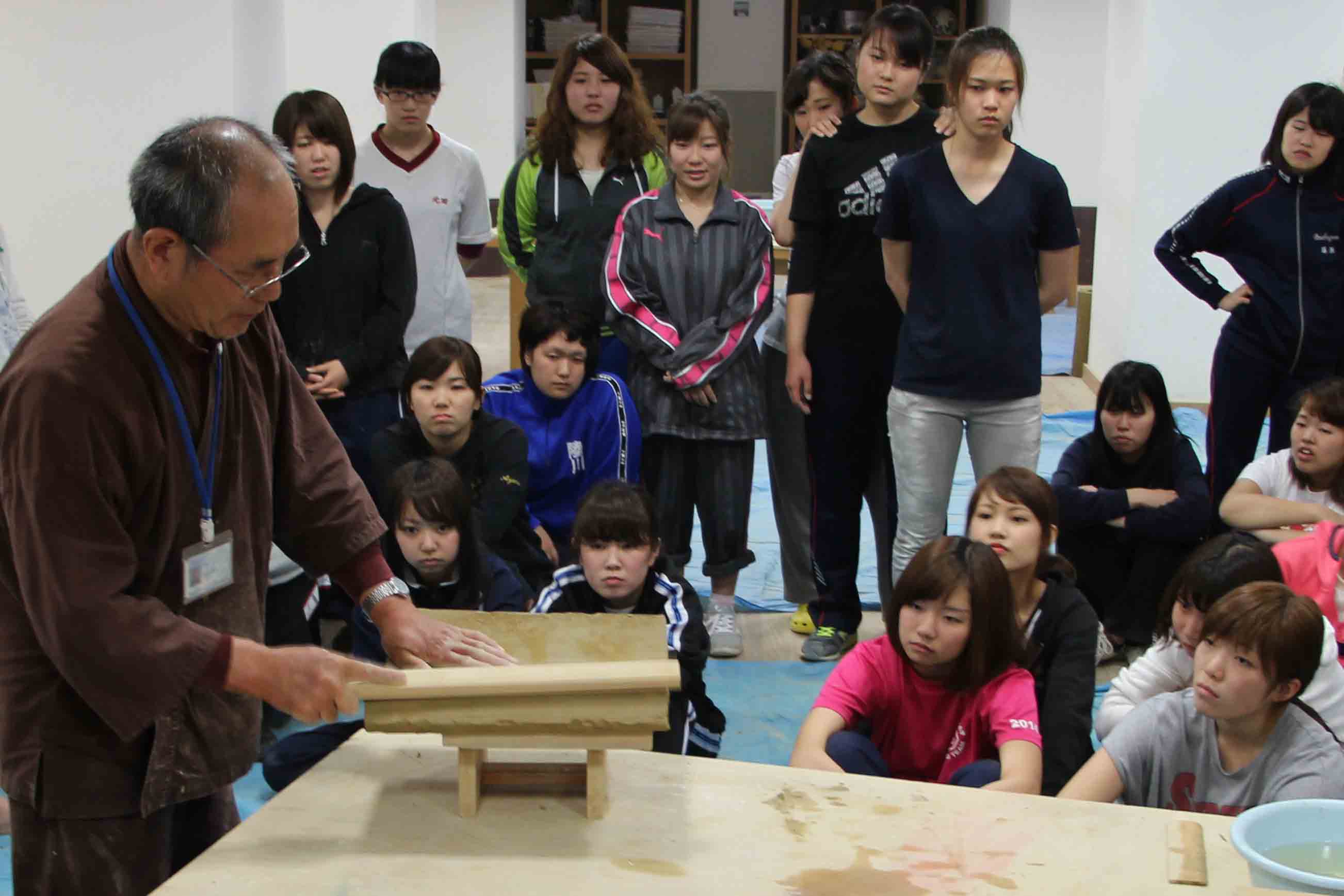} & \rule{1cm}{0cm} & \includegraphics[width=6cm]{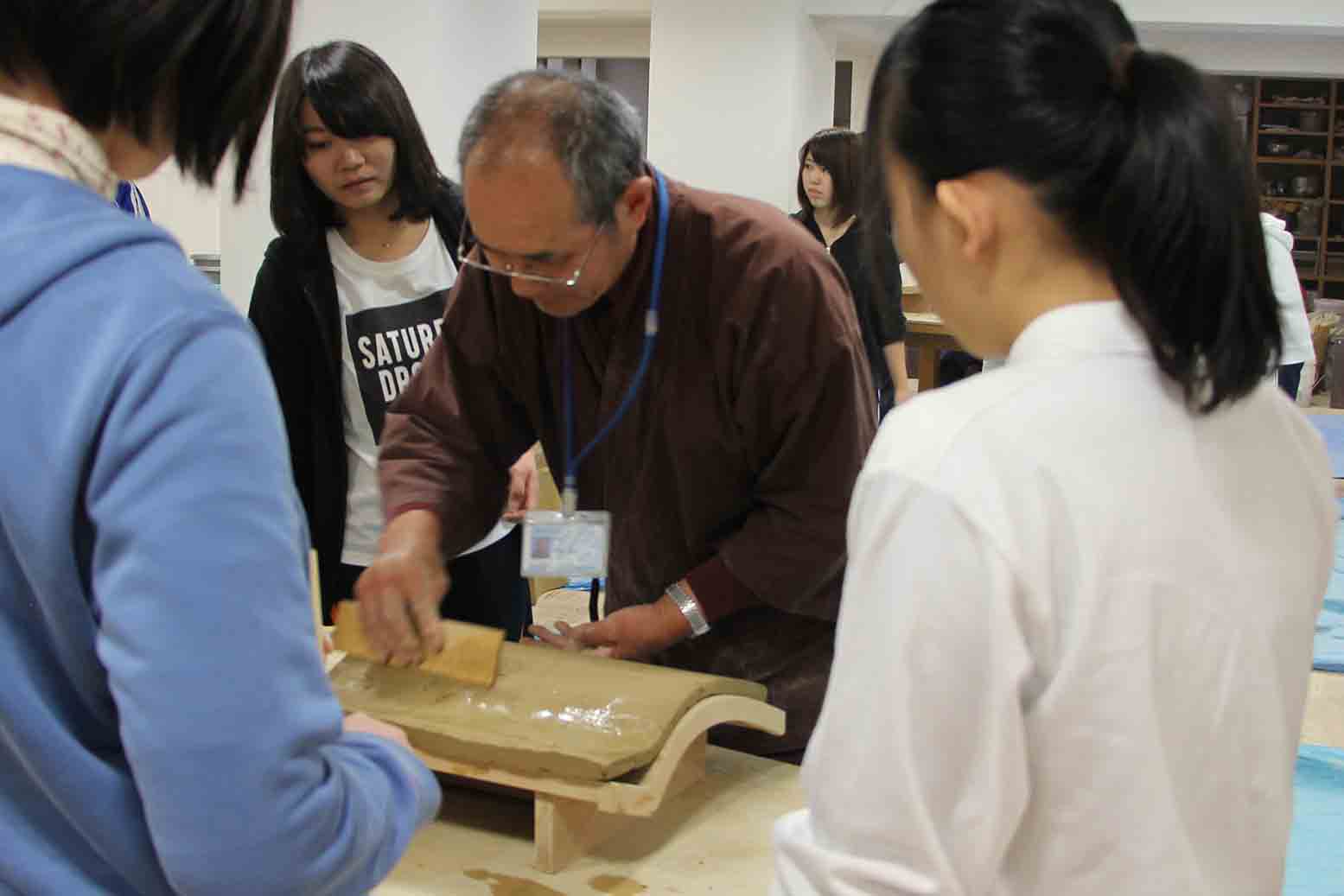}
\end{tabular}
\caption{Beating (left) and stroking (right) in making handmade pantiles in demonstration class at
  Department of Architecture, Mukogawa Women's University.}
\label{fig:sec5.fig6}  
\end{figure}
%
\subsection{Generation of 3D keyline data of pantiles} \label{sec:5.2}
The pantiles we measured had been used in a house built around 1900 in Settsu City, Osaka Prefecture
(Figure \ref{fig:sec5.fig7}). From the characteristic shapes of the pantiles, they were likely used
from the original construction or replaced before the revision of the urban building law in 1924
after the Great Kanto Earthquake. According to the owners, most of the tiles were blown away when
the 2nd Muroto Typhoon hit in 1961, so they were collected and re-roofed. After that, only a few of
the pantiles were replaced before the house was demolished in March 2017.

In our survey before the demolition, we found that the roofs of this house were covered by four
different sizes of pantiles that ranged from 240 to 280 mm in working width. Prior to dismantling
the building, we preserved six rows of pantiles (A to F in Figure \ref{fig:sec5.fig8}) that covered
those four sizes. The pantiles varied in shape due to their handmade nature, so we preserved six
rows instead of only four individual pantiles. We measured 37 pantiles of the C and F with a working
width of 270 mm (the most commonly used on this house), excluding the eave pantiles (C01, F01).
%
\begin{figure}[h]
  \centering
  \begin{tabular}{ccc} \includegraphics[width=6cm]{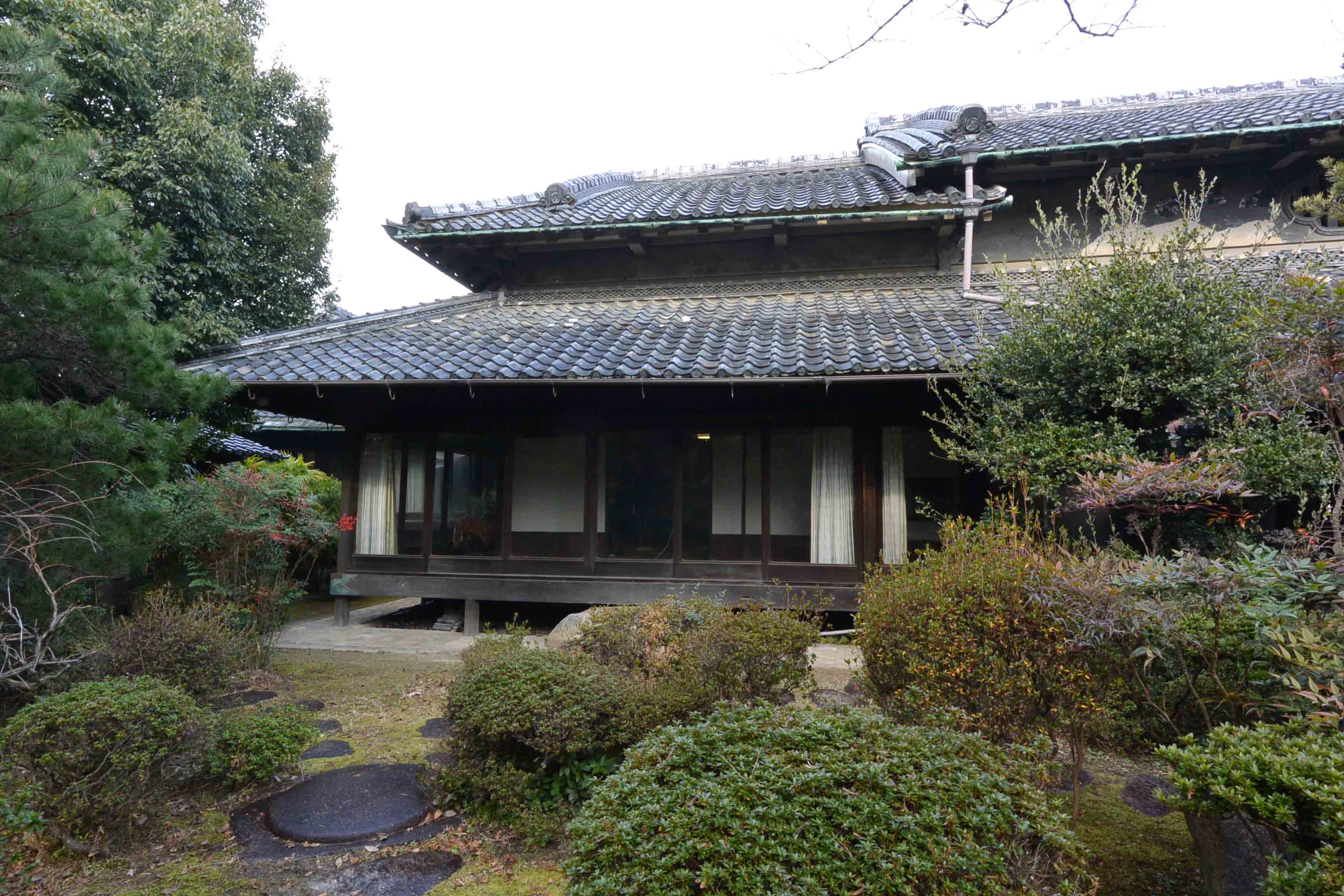} && \includegraphics[width=6cm]{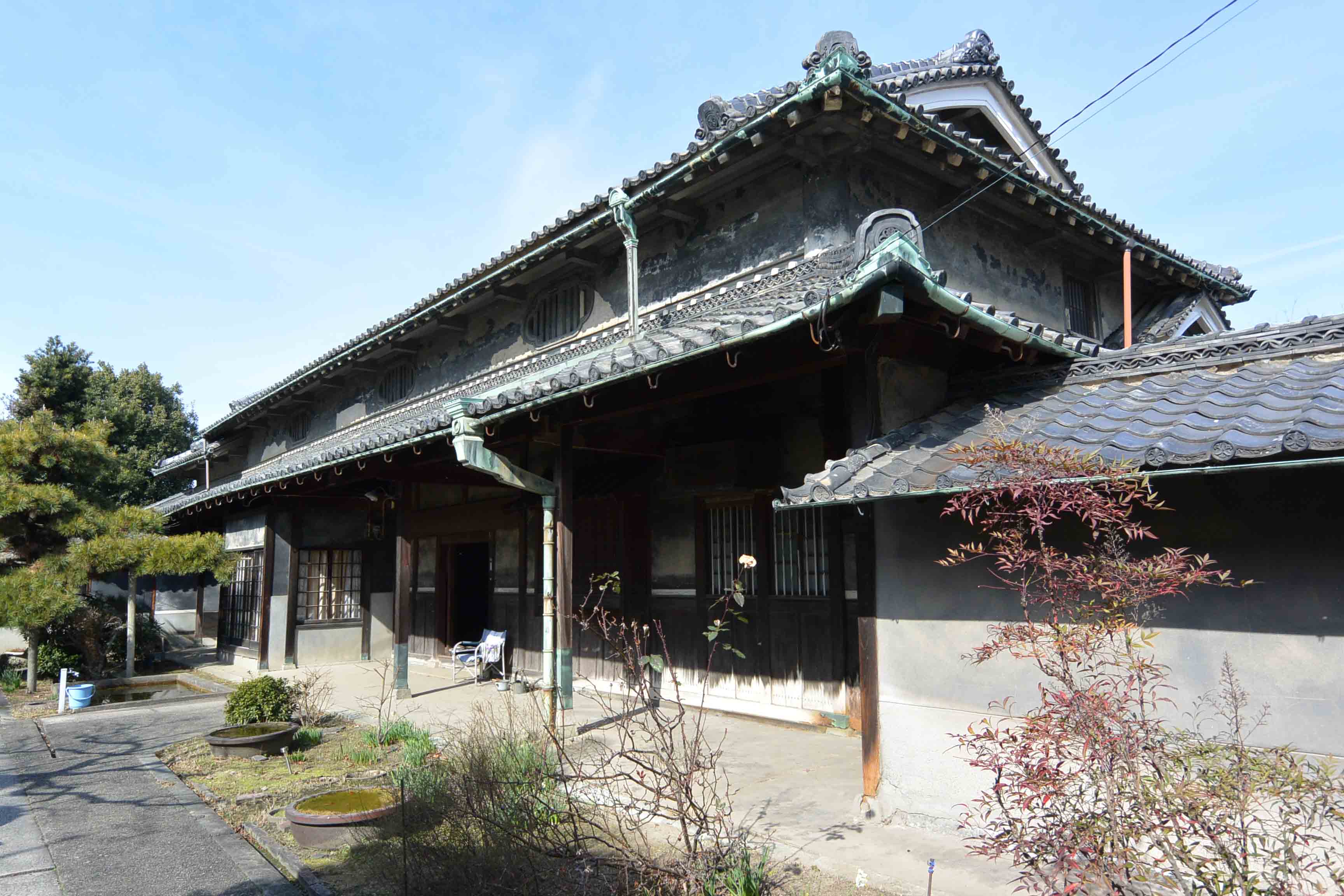}
  \end{tabular}
  \caption{Exterior photos of house with pantile roofs.}
  \label{fig:sec5.fig7}
\end{figure}

\begin{figure}[h]
  \centering \includegraphics[width=14cm]{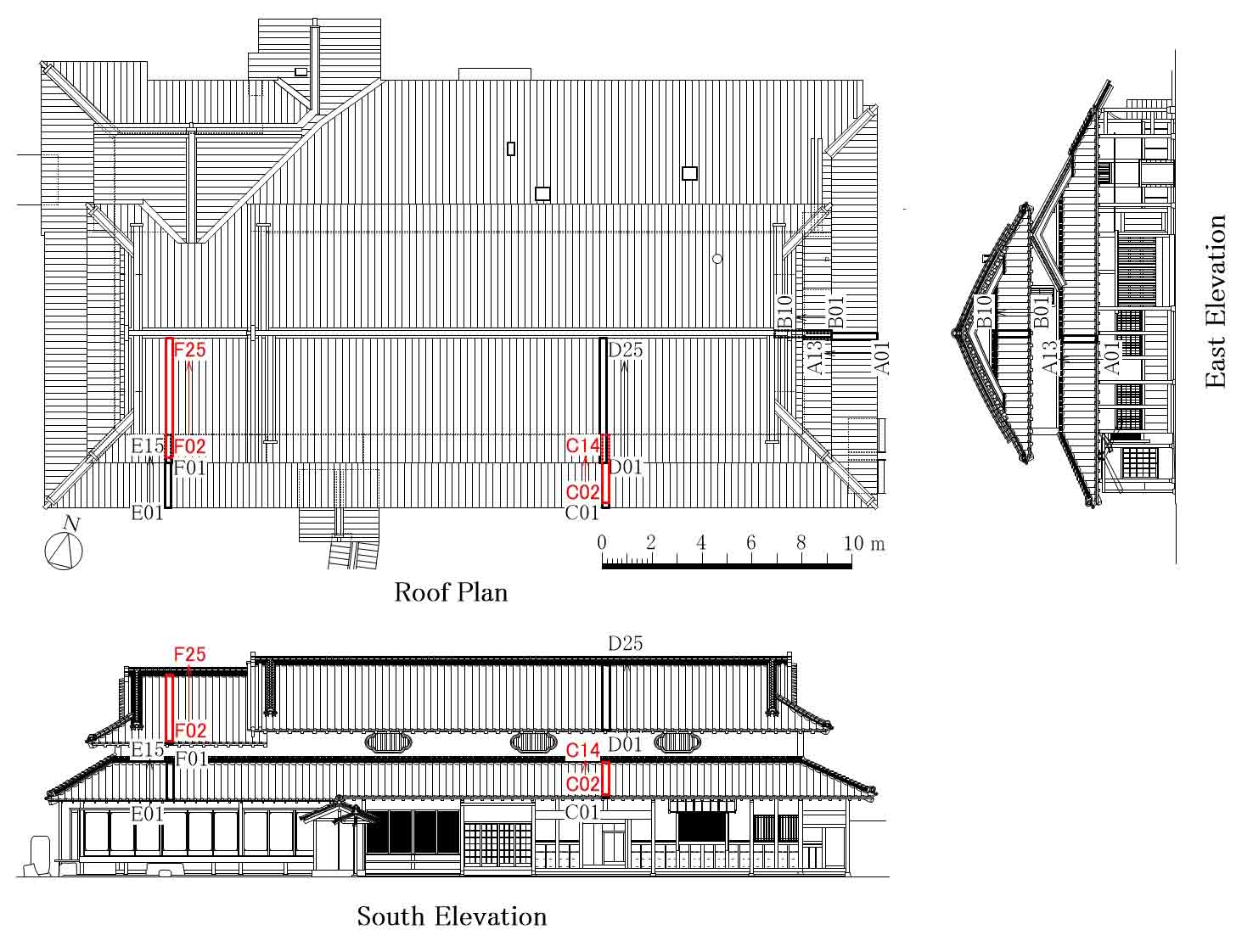}
  \caption{Preserved pantiles (black and red: A-F) and measured pantiles (red: C02-14, F02-25).}
  \label{fig:sec5.fig8}
\end{figure}
%
We used the NextEngine's Ultra HD 3D laser scanner and placed it
and the pantiles as shown in Figure \ref{fig:sec5_09} for the measurements.
%
\begin{figure}[h!]
  \centering \includegraphics[width=14cm]{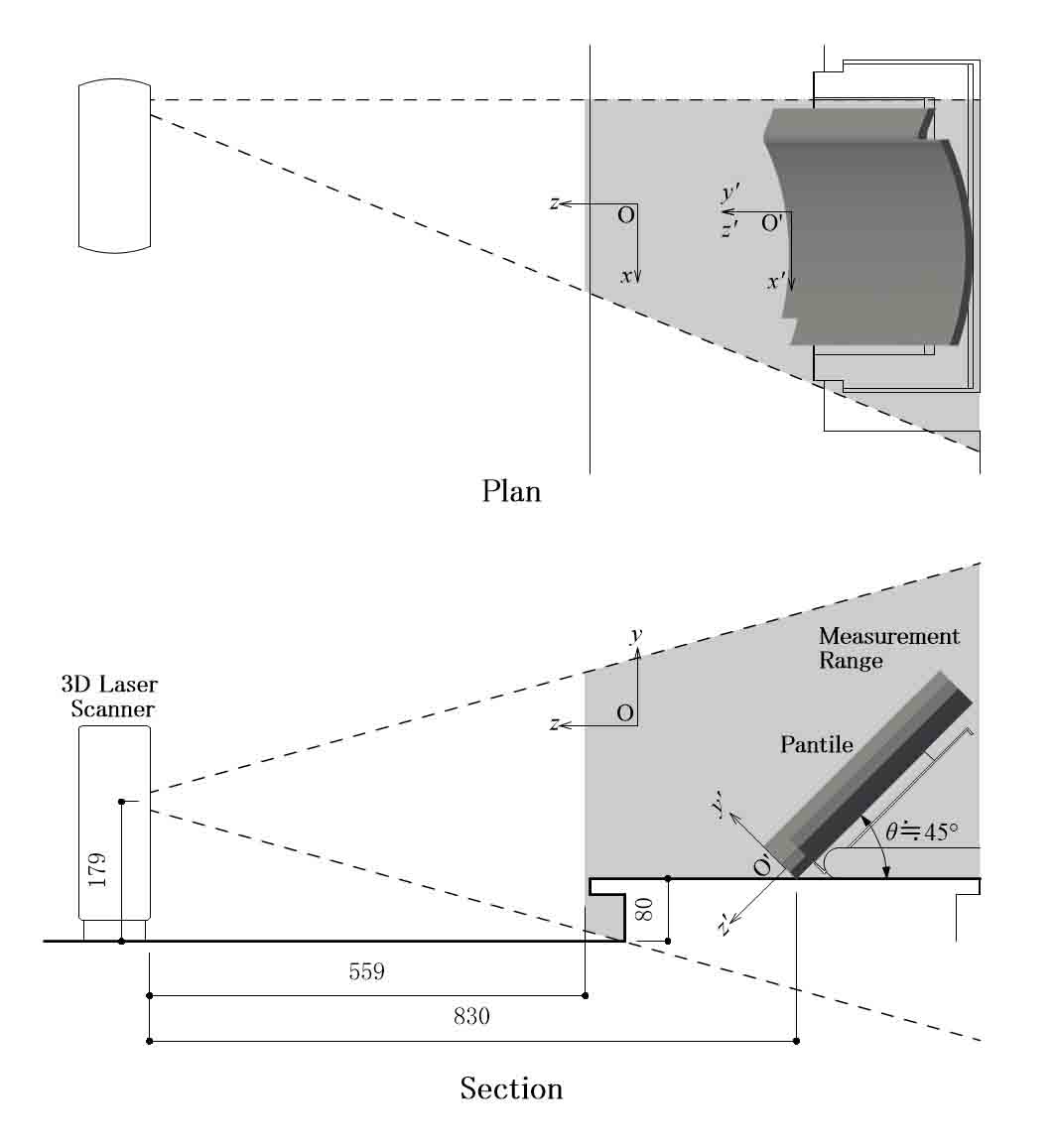}
  \caption{Placement of 3D laser scanner and pantile during measurement.}
  \label{fig:sec5_09}
\end{figure}
%
The mesh data obtained from the 3D laser scanner was read by 3D Systems' RapidWorks 64 4.1.0 reverse
modeling software. We used this software to synthesize and decimate its polygons finer than the
scanner's measurement accuracy (0.3 mm) and then automatically heal incorrect data and fill holes in
the meshes.

The two front edges are the 3D keylines of the pantile that could be observed when pitched. The
lower edge could be generated by extracting the outer boundary curve of the mesh with RapidWorks.
The upper edge could be generated by extracting the curve network from the mesh data with
RapidWorks. Figure \ref{fig:sec5_10} shows an example of the two 3D keylines (upper and lower edges)
generated. Each of these 3D keylines form an open polygon.
%
\begin{figure}[h!]
  \centering \includegraphics[width=6cm]{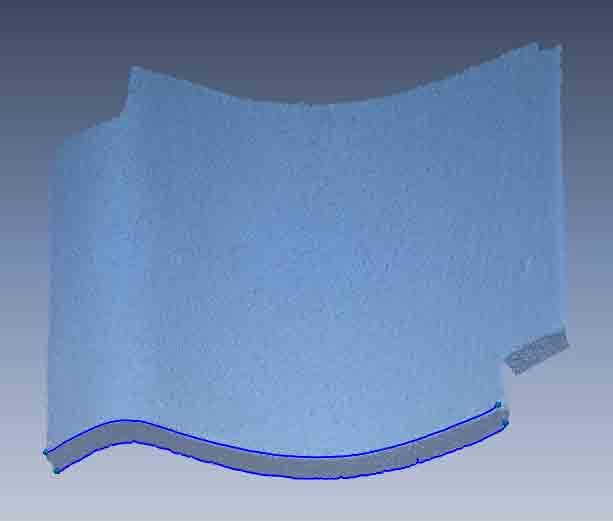}
  \caption{Example of 3D keylines generated (C07).}
  \label{fig:sec5_10}
\end{figure}
%
\subsection{Generation and conversion of 2D keylines} \label{sec:5.3}
%
\paragraph{Plane fitting of 3D keyline by principal component analysis}\hfill\\
In order to generate 2D keyline data from the 3D keylines obtained from the scanning data, we need
plane fitting of the 3D data. This can be done by projecting the points in $\mathbb{R}^3$ on the
plane which minimizes the sum of squared distances from the points. As is well-known, such plane is
constructed by applying the principal component analysis of the so-called the covariance matrix.
More concretely, let $p_1, p_2, \dots, p_m\in \mathbb{R}^3$ be three-dimensional column vectors that
represent spatial coordinates of a point cloud consisting of $m$ points. Their center of
gravity $\bar{p}$ is given by
\begin{equation} \label{eq:centerg}
  \bar{p} = \frac{1}{m} \sum_{k=1}^m p_k.
\end{equation}  
We consider the covariance matrix $C$ defined by
\begin{displaymath}
C = \frac{1}{m} \sum_{k=1}^m (p_k - \bar{p}) \tr{(p_k - \bar{p})},
\end{displaymath}
which is real, symmetric positive semi-definite matrix, and let $e_1, e_2, e_3$ be orthonormal
eigenvectors that correspond to eigenvalues $\lambda_1, \lambda_2, \lambda_3$ ($\lambda_1 \geq
\lambda_2 \geq \lambda_3\geq 0$) of $C$.  Then it is known that the plane $\pi$ that includes
$\bar{p}$ and is parallel to $e_1$ and $e_2$ (i.e., with $e_3$ as the normal vector) minimizes the
sum of the squared distances from $p_1, p_2, \dots, p_m$ (see, for example, \cite{Berkmann-Caelli:1994}). 
Therefore, plane fitting is obtained by projecting $p_1, p_2, \dots, p_m$ onto $\pi$. 

Following this idea, the $xyz$-coordinates of the point sequence consisting of all the vertices of
the open polygon that constitutes one 3D keyline are converted into the $x'y'z'$-coordinate system,
where the center of gravity of the point sequence is the origin $O'$ and the orientations of $e_1$,
$e_2$, and $e_3$ are the $x'$, $y'$ and $z'$ axes, respectively. The orientations of $e_1$, $e_2$,
and $e_3$ are determined to be those of the $x'$, $y'$ and $z'$ axes shown in Figure
\ref{fig:sec5_09}, respectively. When converting the coordinate $p_k$ in the $xyz$-coordinate system
into the coordinate $p'_k$ in the $x'y'z'$-coordinate system,
\begin{displaymath}
p'_k = \tr{(e_1\,e_2\,e_3)}(p_k - \bar{p}),
\end{displaymath}
is satisfied because $(e_1\,e_2\,e_3)$ is an orthogonal matrix.
%
\paragraph{Equilateral open polygon approximation of 2D keyline}\hfill\\
The 2D keylines obtained by projecting the 3D keylines onto the $x'y'$ planes in the previous
paragraph (simply referred to as the 2D keylines before approximation) can be approximated by
equilateral open polygons using the following procedure.

First, let the end point on the negative side of the $x'$ axis of the 2D keyline before
approximation (simply referred to as the left end) be the starting point $A_1$ after
approximation, and let the point on the keyline whose distance from $A_1$ is $r$ (= 1.0 mm) be
$A_2$.  Next, let the point on the same keyline where the distance from $A_2$ is $r$ and does not
return to the left end ($A_1$) be $A_3$.  In the same way, repeat the operation to let the point
whose distance from $A_k$ is $r$ and does not return to $A_{k-1}$ be $A_{k+1}$ (Figure
\ref{fig:sec5_11}).  However, if there are multiple points that satisfy this condition, let the
point on the keyline closer to the right end (the end point on the positive side of the $x'$ axis)
be $A_{k+1}$ (Figure \ref{fig:sec5_12}).  When the only point that has a distance $r$ from $A_n$ is
the point that returns to $A_{n-1}$, $A_n$ is the stop point of the 2D keyline after approximation
(Figure \ref{fig:sec5_13}).  All the points from the start point $A_1$ to the stop point $A_n$ are
connected on the sides to make the 2D keyline after approximation.  The 2D keyline after
approximation is an equilateral open polygon with $n$ vertices, $n-1$ sides, and $r(n-1)=n-1$ mm in
total length.
%
\begin{figure}[h!] \centering
  \begin{minipage}[t]{0.25\textwidth} \centering \includegraphics[width=5cm]{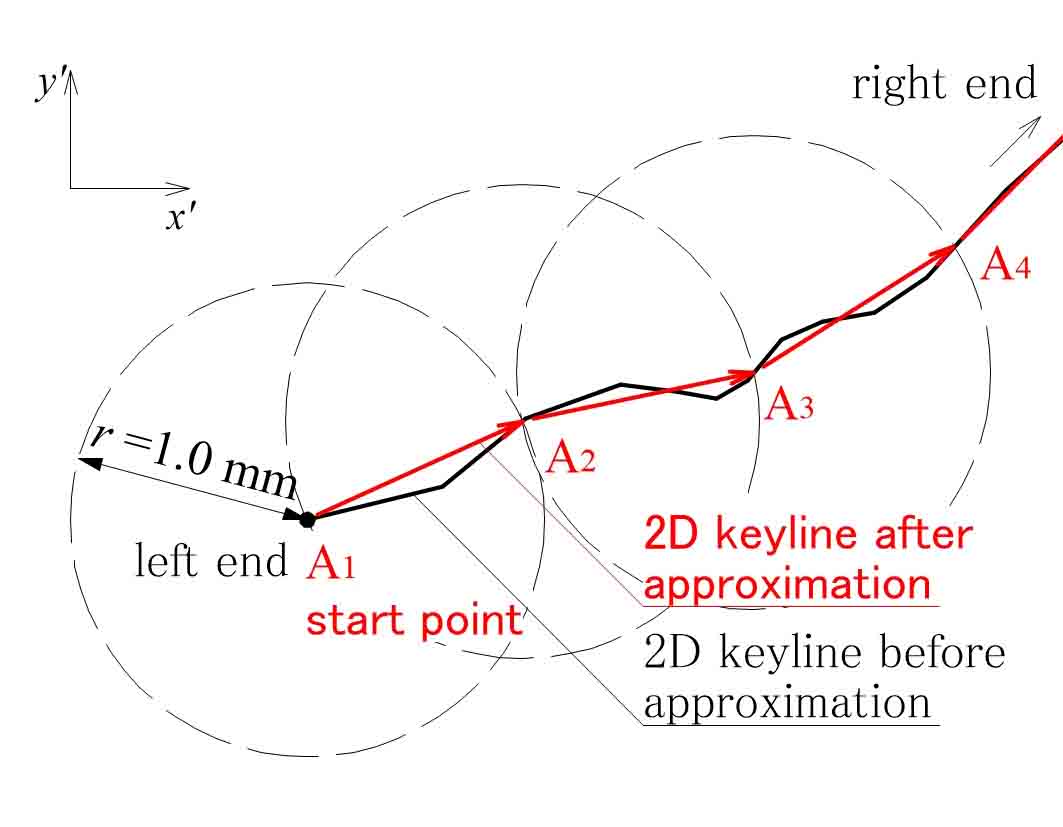}
    \caption{Procedure for equilateral open polygon approximation of 2D keyline (C07 lower
keyline).}
    \label{fig:sec5_11}
  \end{minipage} \qquad
  \begin{minipage}[t]{0.25\textwidth} \centering \includegraphics[width=5cm]{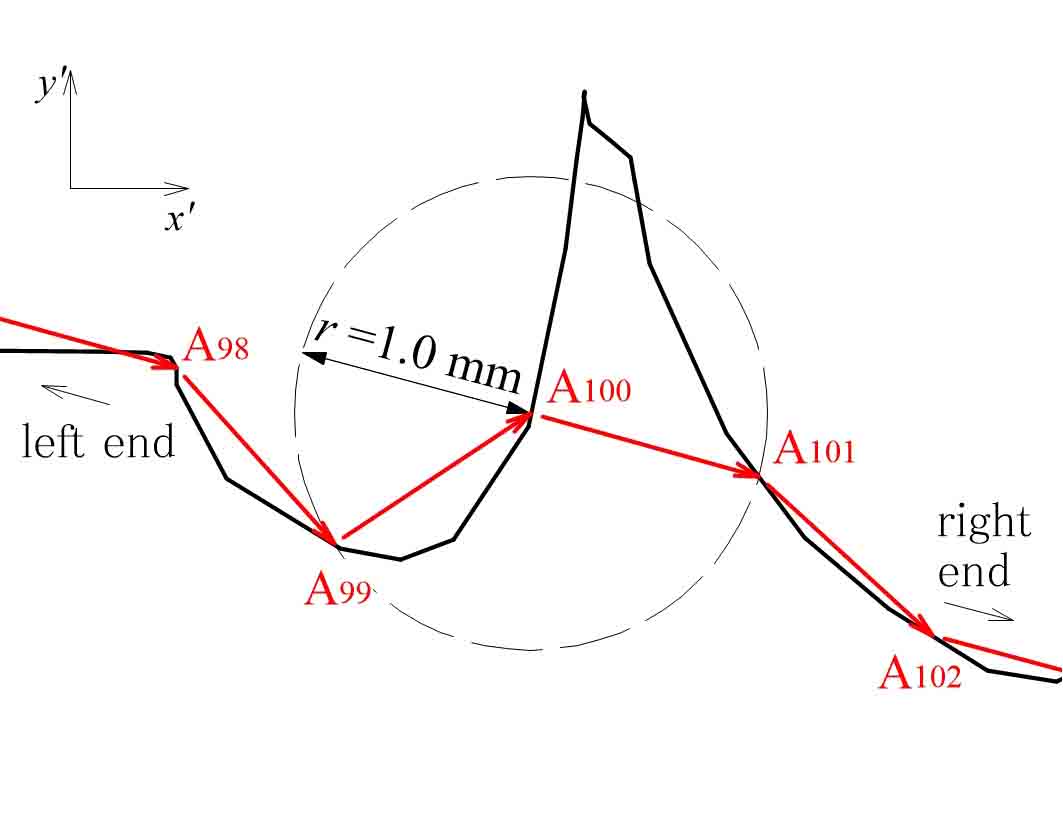}
    \caption{Method of determining vertex of equilateral open polygon when there are multiple points
whose distance is r (C07 lower keyline).}
    \label{fig:sec5_12}
  \end{minipage} \qquad
  \begin{minipage}[t]{0.25\textwidth} \centering \includegraphics[width=5cm]{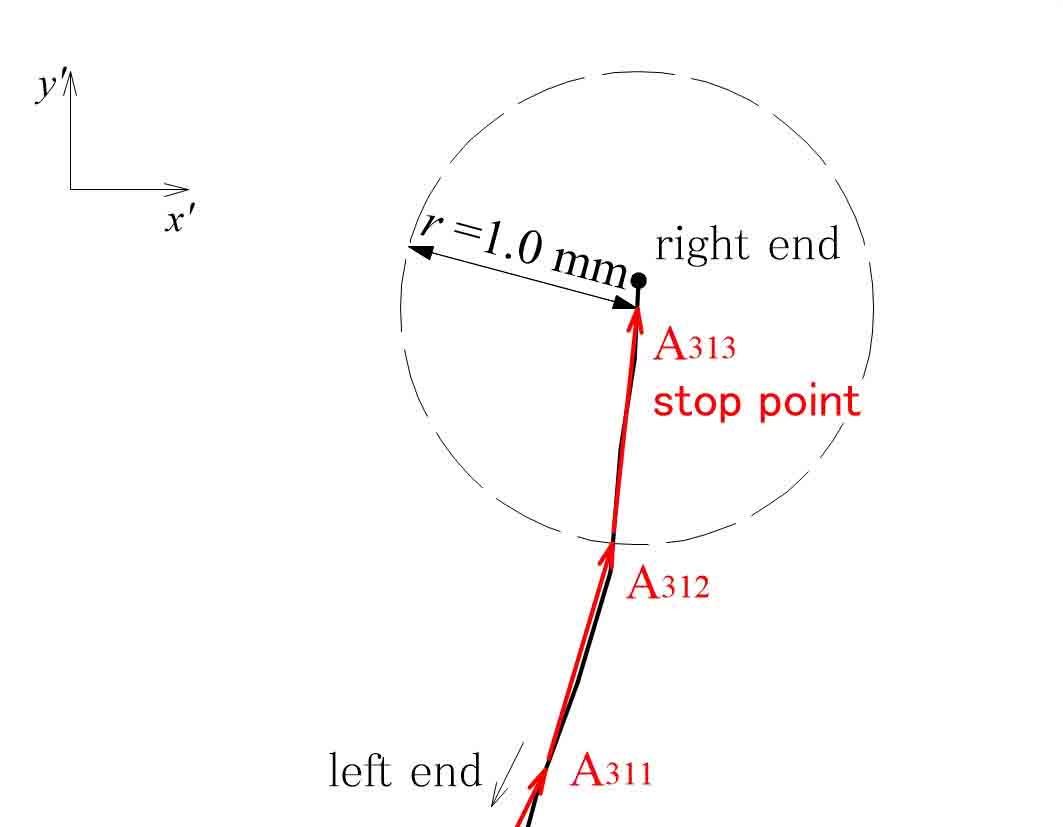}
    \caption{Method of determining stop point of equilateral open polygon (C07 lower keyline).}
    \label{fig:sec5_13}
  \end{minipage}
\end{figure}
%
\paragraph{Coordinate conversion of 2D keylines after approximation by principal component analysis}\hfill\\
When measuring the pantiles as shown in Figure \ref{fig:sec5_09}, we place them one by one on the
table by hand. Therefore, the position and rotation angle of each pantile is slightly different. To
compare the curves of the keylines of the measured pantiles, we have to eliminate the effects of the
positions and rotation angles and convert them into a coordinate system determined by only the
keyline shape. This effect is already eliminated by conversion from the $xyz$-coordinate system into
the $x'y'$-coordinate system, where the center of gravity of the point sequence consisting of all
the vertices of the 2D keyline is $O'$, the orientation of its first principal component is the
$x'$ axis, and the orientation of its second principal component is the $y'$ axis. However, the
number and position of the points that compose the point sequence have been changed due to the
equilateral open polygon approximation. Therefore, the center of gravity of the point sequence
consisting of all vertices of the approximated 2D keyline does not generally coincide with $O'$, nor
does the orientation of its first principal component coincide with the $x'$ axis, nor does the
orientation of its second principal component coincide with the $y'$ axis. Because the relationship
between the shape of the keyline and the coordinate system has been lost, the coordinate system must
be converted again into one determined solely by the shape of the keyline.  Therefore, by performing
the principal component analysis again on the point sequence consisting of all the vertices of the
2D keyline after approximation, we convert the coordinates into the $x''y''$-coordinate system with
the center of gravity of the point sequence as $O''$ and the first principal component as the $x''$
axis.  In the $x'y'$-coordinate system, let $q_1,q_2,\dots,q_n$ be the two-dimensional column
vectors that represent the plane coordinates of the $n$ points of the 2D keyline after
approximation. Then their center of gravity $\bar{q}$ is given by
\begin{displaymath}
\bar{q} = \frac{1}{n}\sum_{k=1}^n q_k. 
\end{displaymath}
Let $f_1,f_2$ be normalized eigenvectors that correspond to the eigenvalues $\mu_1,\mu_2$
($\mu_1 \geq \mu_2$) of covariance matrix $D$ defined by
\begin{displaymath}
 D = \frac{1}{n} \sum_{k=1}^n (q_k - \bar{q}) \tr{(q_k - \bar{q})}.
\end{displaymath}
Here, the directions of $f_1$ and $f_2$ are determined so that the diagonal components of
$(f_1\, f_2)$ are all positive.  When converting the coordinate $q_k$ in the $x'y'$-coordinate
system into the coordinate $q'_k$ in the $x'y'z'$-coordinate system,
\begin{displaymath}
 q'_k = \tr{(f_1\, f_2)} (q_k - \bar{q}),
\end{displaymath}
is satisfied because $(f_1\, f_2)$ is an orthogonal matrix.

In the following, the $x''y''$-coordinates obtained by this coordinate conversion will be redefined
as $xy$-coordinates.
%
\paragraph{Estimation of inflection point of approximate curve of 2D keyline}\hfill\\
The approximate curve of the 2D keyline of the pantile has one global inflection point and is
clearly asymmetric across the inflection point.  In general, it is difficult to formulate such a
curve in terms of a single elastica.  Therefore, we estimate the global inflection point of the
approximate curve, then we divide the curve on both sides of the inflection point and approximate
them with different elastic curves.  We estimate the inflection point using a method inspired by the
Ramer-Douglas-Peucker (RDP) algorithm \cite{douglas,ramer}, which was devised to simplify
an open polygon with many vertices by thinning out the vertices (Figure \ref{fig:sec5_14}).  We call
this the RDP method:
\begin{enumerate}
  \item Among the vertices of the open polygon, let the point farthest from the line segment
    connecting the two end points be 3.  Let the end point on the side with the inflection point
    (viewed from point 3) be 2, and let the end point on the side without the inflection point be 1.
  \item Among the vertices of the open polygon between points 2 and 3, let the point farthest from
    line segment 2-3 be 4.
  \item Among the vertices of the open polygon between the points $P$ ($P=3,4,5,\dots$) and $P+1$,
    repeat the calculation to let the point farthest from the line segment $P$($P+1$) be $P+2$.
    However, the numbers of point $P$ and point $P+1$ are switched and the next calculation is
    performed if point $P+1$ is the upper side of line segment ($P-1$)$P$ and point $P+2$ is also
    the upper side of line segment $P$($P+1$), or, conversely, if point $P+1$ is the lower side of
    line segment ($P-1$)$P$ and point $P+2$ is also the lower side of line segment $P$($P+1$).
  \item When points $P+1$ and $P+2$ become adjacent vertices on the open polygon, the calculation is
    terminated and the inflection point is estimated to be $P+2$.
\end{enumerate}
%
\begin{figure}[h!]
  \centering \includegraphics[width=12cm]{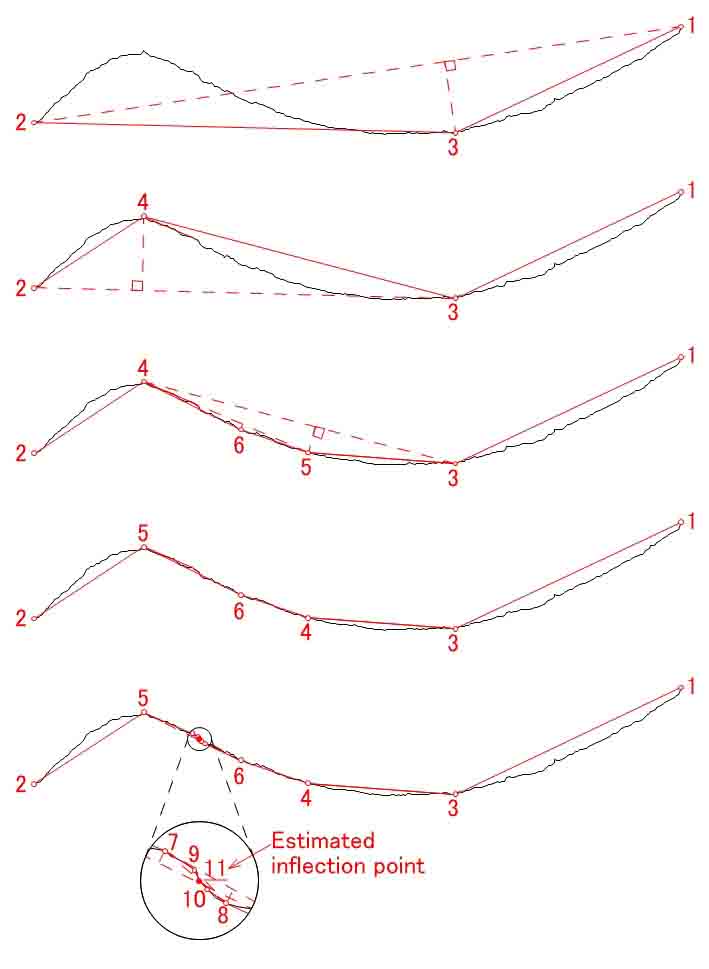}
  \caption{Example of inflection point estimation by RDP method (C07 lower keyline). In the first three curves we apply steps (1), (2) and (3), respectively. In the fourth and fifth curves, following step (3), we switch point 4 to point 5, and point 7 to point 11, respectively. Finally, in the augmented image at the bottom, following step (4), the inflection point is estimated to be point 11.}
  \label{fig:sec5_14}
\end{figure}
%
Because the RDP algorithm is intended to thin out the vertices, its calculation is terminated when an
error falls within a certain range.  However, the RDP method in this study is meant to estimate the
inflection points.  Therefore, we repeat the calculation only in the section where the inflection
point is estimated to exist until it agrees with the original open polygon, and we estimate the
converged point to be the inflection point.
%
\subsection{Approximation of 2D keyline by discrete elastica} \label{sec:5.4}
For the right (valley) sides of the inflection points of the 37 lower keylines (Figure
\ref{fig:sec5_15}) of the 37 pantiles, we tried to approximate the discrete elastica by the method
already described.  Among $p=(x_0,y_0,h,\phi,z,q,k)$ shown in \eqref{eqn:param}, $(x_0,y_0 )$ were
set as the inflection point to allow the discrete elastica to pass through the inflection point
after approximation, and the remaining five parameters were optimized.  An example of approximation
of a keyline of a pantile (C07 lower keyline) by discrete elastica is shown in Figure
\ref{fig:sec5_16}, the calculation results of the 37 approximated discrete elasticae are shown in Figure
\ref{fig:sec5_17}, and the calculation results of the parameters $h$, $\phi$, $z$, $q$, $k$ are shown in
Table \ref{table:results}.
%
\begin{figure}[h!]
  \centering \includegraphics[width=12cm]{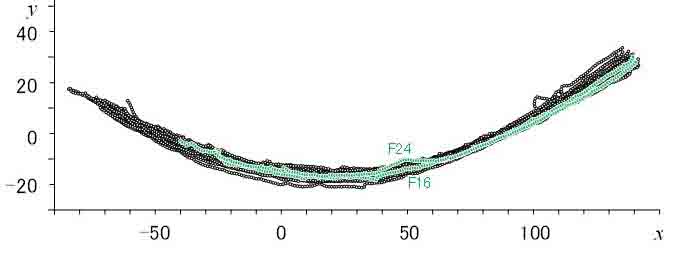}
  \caption{Right (valley) sides of 2D lower keylines for calculation.}
  \label{fig:sec5_15}
\end{figure}

\begin{figure}[h!]
  \centering \includegraphics[width=12cm]{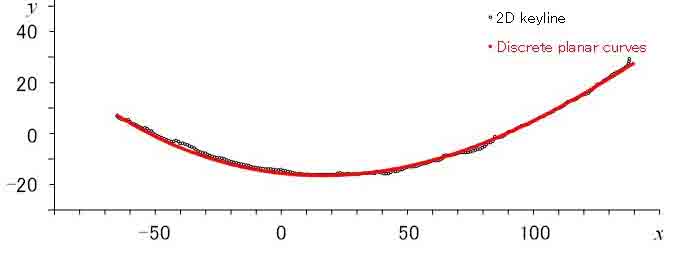}
  \caption{Example approximation by discrete elastica (C07 lower keyline).}
  \label{fig:sec5_16}
\end{figure}

\begin{figure}[h!]
  \centering \includegraphics[width=12cm]{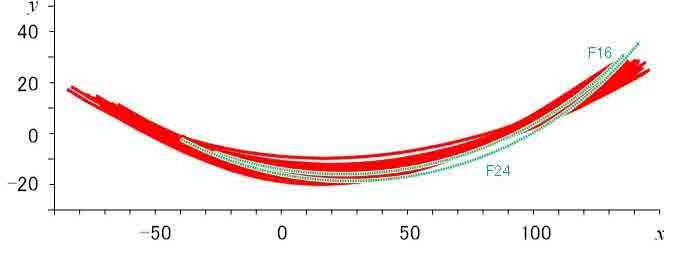}
  \caption{Calculation results of approximated discrete elasticae of 2D lower keylines.}
  \label{fig:sec5_17}
\end{figure}
%
\begin{table}[ht]
  \centering
  \resizebox{\textwidth}{!}{%
  \begin{tabular}{ll|c|cc|c|c|c}
    \toprule \multicolumn{2}{c|}{
    \multirow{2}{*}{Lower keylines right (valley) side}} & $h$ & \multicolumn{2}{c|}{$\phi$} & \multirow{2}{*}{$z$} & \multirow{2}{*}{$q$} & \multirow{2}{*}{$k$}\\
                                                         && [mm] & [rad] &[grad] &&&\\
    \midrule \multirow{2}{*}{{C02--14, F02--15, 17--23, 25}} & mean & $0.964$ & $-0.168$ & $-9.623$ & $9.001\times10^{-3}$& $5.919$ & $0.353$\\
                                                         &std. dev.& $6.564\times10^{-3}$ & $2.078\times10^{-3}$ & $0.119$ & $1.263\times10^{-3}$& $1.789\times10^{-3}$& $1.235\times10^{-2}$\\
    \hline
    \multicolumn{2}{l|}{F16} & $0.945$ & $2.045$ & $117.192$ & $67.36$ & $5.226$ & $27.43$\\
    \hline
    \multicolumn{2}{l|}{F24} & $1.007$ & $2.478$ & $141.986$ & $121.8$ & $28.28$ & $27.39$\\
    \bottomrule
  \end{tabular}
  }
  \caption{Calculation results for $h$, $\phi$, $z$, $q$ and $k$ (mean and standard
    deviation). Because the values for F16 and F24 were extremely different from the others, the
    means and standard deviations were calculated for the 35 keylines except for F16 and F24, and
    the values for F16 and F24 were written separately.}
  \label{table:results}
\end{table}
%
As Table 1 shows, the 35 key lines (except for F16 and F24) show very little variation in the
calculated results of the parameters, and the discrete elasticae are similar in both shape and
rotation angle.  The $k$ values are close to $0.3$, and the shapes are close to sine curves.
However, the values of $\phi$, $z$, $q$ and $k$ for F16 and F24 are extremely different from the
others. The $k$ values exceed 10, and the shapes are close to arcs.

Figure \ref{fig:sec5_18} shows the 37 translated discrete elasticae shown in Figure
\ref{fig:sec5_17} with the inflection point $(x_0, y_0 )$ at the origin.  The variations of the
discrete elasticae are larger than that in Figure \ref{fig:sec5_17}, suggesting that a certain
number of variations and errors may be included in the estimations of the inflection points.  In F16
and F24 in particular, the estimated inflection points are located closer to the right (valley)
sides. Therefore, we can assume that the areas around the inflection points where the curvature is
small are largely omitted, and the discrete elasticae are approximated to be close to the arcs.  The
effect the local unevenness of the keyline has on the estimation of the inflection point should be
examined.
%
\begin{figure}[h!]
  \centering
  \includegraphics[width=12cm]{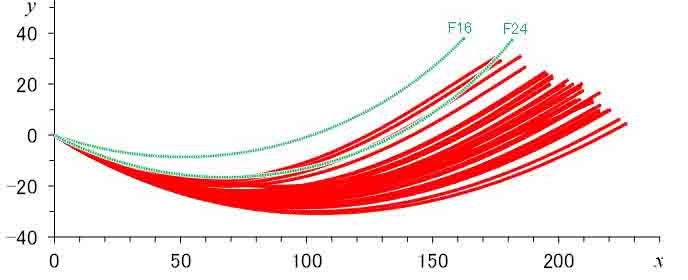}
  \caption{Result of translating discrete elasticae in Figure 17 so that each inflection point is at
    origin. Calculation results of F16 and F24 are clearly quite different from others.}
  \label{fig:sec5_18}
\end{figure}
%
In conclusion, we found that 35 of the 37 lower keylines of the handmade pantiles could be
approximated by discrete elasticae with a very small variation on the right (valley) side of the
inflection point.  However, the errors of the estimated positions of the inflection points may
affect the accuracy of the approximations.
%
\section*{Acknowledgements} 
The authors would like to thank Professor David Brander for encouragement. They also express their
thanks to Professors Nozomu Matsuura, Jun-ichi Inoguchi, Kenjiro T. Miura, Satoshi Kanai and Dr.
Masahisa Asada for fruitful discussions. They would like to thank Mr. Konosuke Onishi, Mr. Yunosuke
Onishi, Associate Professor Toshikazu Inoue, and Mr. Daisuke Mitsumoto for their cooperation in the
housing survey. This work was initiated by the 2018 IMI Joint Use Research Program Short-Term Joint
Research No.20180008, and supported by JSPS Kakenhi JP16H03941, 17K00741, 20K12520, 21K03329 and JST
CREST Grant Number JPMJCR1911.
%

\end{document}